\documentclass[titlepage,12pt]{article}

\usepackage[dvips]{graphicx}

\usepackage[myheadings]{fullpage}
\usepackage{pmetrika}

\usepackage{submit}
\usepackage{apacite}

\RequirePackage{amsmath,amsfonts,amssymb}
\RequirePackage[authoryear]{natbib}
\RequirePackage[colorlinks,citecolor=blue,urlcolor=blue]{hyperref}
\RequirePackage{graphicx}

\usepackage{amsmath,natbib}
\usepackage{multirow}

\usepackage{graphicx,bm,mathtools,amsfonts,bbm,hyperref,enumitem}
\numberwithin{equation}{section}

\usepackage{algorithmicx}
\usepackage{algorithm}
\usepackage[noend]{algpseudocode}
\algnewcommand\algorithmicinput{\textbf{Input:}}
\algnewcommand\Input{\item[\algorithmicinput]}
\algnewcommand\algorithmicoutput{\textbf{Output:}}
\algnewcommand\Output{\item[\algorithmicoutput]}

\newcommand{\bY}{\textbf{Y}}
\newcommand{\bX}{\textbf{X}}

\newcommand{\bZ}{\textbf{Z}}
\newcommand{\bR}{\textbf{R}}
\newcommand{\btheta}{\bm{\theta}}
\newcommand{\bbeta}{\bm{\beta}}
\newcommand{\bgamma}{\bm{\gamma}}
\newcommand{\balpha}{\bm{\alpha}}

\DeclareMathOperator*{\argmax}{argmax}

\newcommand{\E}{\mathbb{E}}

\usepackage{tikz}
\usetikzlibrary{shapes, arrows, positioning, fit}
\tikzstyle{terminator} = [rectangle, draw, text centered, rounded corners, minimum height=4em]

\usepackage{soul}
\setstcolor{magenta}

\renewcommand{\hat}[1]{\widehat{#1}}

\renewcommand{\tilde}[1]{\widetilde{#1}}

\begin{document}

\begin{titlepage}

\title{}


\author{Daniel Suen}

\affil{University of Washington}

\author{Yen-Chi Chen}
\affil{University of Washington}

\vspace{\fill}\centerline{\today}\vspace{\fill}

\comment{Competing interests: The authors declare none.}

\comment{Financial support: DS was supported by NSF DGE-2140004.
	YC was supported by NSF DMS-195278, NSF DMS-2112907, NSF DMS-2141808, and NIH U24-AG072122.}
\linespacing{1}
\contact{Correspondence should be sent to\\

\noindent E-Mail: dsuen@uw.edu \\
\noindent Phone: 1-206-486-0446
}

\end{titlepage}

\setcounter{page}{2}
\vspace*{2\baselineskip}

\RepeatTitle{Modeling missing at random neuropsychological test scores using a mixture of binomial product experts}\vskip3pt

\linespacing{1.5}
\abstracthead
\begin{abstract}
Multivariate bounded discrete data arises in many fields.  In the setting of dementia studies, such data is collected when individuals complete neuropsychological tests.  We outline a modeling and inference procedure that can model the joint distribution conditional on baseline covariates, leveraging previous work on mixtures of experts and latent class models.  Furthermore, we illustrate how the work can be extended when the outcome data is missing at random using a nested EM algorithm.  The proposed model can incorporate covariate information and perform imputation and clustering.  We apply our model on simulated data and an Alzheimer's disease data set.

\end{abstract}\vspace{\fill}\pagebreak

\section{Introduction}

\label{sect:intro}

Modeling multivariate discrete data is a common problem across many fields such as social sciences, psychology, and ecology.  For instance, in education research, discrete data arises from students' test scores \citep{maydeu-olivares2015item}, and in ecology, discrete data arises when the counts of a given animal population are measured over various areas and time periods \citep{Anderson2019, Maslen2023}.  As the size of data sets continues to expand, missing data becomes increasingly prevalent \citep{kang2013prevention}, and it is common to have access to a rich set of discrete variables that is subject to missingness.  Both the multivariate discreteness and the missingness alone can create modeling difficulties, but together, they present a unique challenge for statistical inference.

\subsection{The National Alzheimer's Coordinating Center database}

Although the general problem statement is widely important, our work is primarily motivated by the analysis of neuropsychological test scores in the database of the National Alzheimer's Coordinating Center (NACC)\footnote{https://naccdata.org/}.
The NACC is an NIH/NIA-funded  center that maintains the largest longitudinal database for Alzheimer's disease in the United States. The NACC coordinates 33 Alzheimer's Disease and Research Centers (ADRC) in the United States.

In the NACC data, we have 41,181 individuals with data collected from 2005 to 2019. 
In this study, we are particularly interested in information regarding
the neuropsychological tests. 
The neuropsychological tests are a set of examinations measuring an
individual's cognitive ability
in four different domains: language, attention, executive function, and memory. 
The scores from these tests are important features for dementia research because
they are not based on a clinician's judgment.
Similar to a conventional exam, the outcome of a neuropsychological test 
is a discrete number taking integer values and has a known range. 
Note that different tests have a different ranges.
In our study, we consider eight neuropsychological test scores per individual. 

A goal of this paper is to introduce a simple and interpretable statistical model
	for modeling the neuropsychological tests. 
	The discrete and correlated nature of the neuropsychological tests presents
	a challenge for modeling them. As far as we are aware, these is no statistical model
	for handling these test scores directly, not to mention the additional challenges
	from missingness of the test scores. 
	To simplify the problem, we focus on a static model that only 
	considers the test scores from each individual's initial visit.
	While there are many covariates that could be incorporated in the analysis,
	we choose to include only four demographic variables (age, education, sex, and race)
	to reduce the complexity of the model.
	These variables are key demographic variables that are often
	included in any dementia research.

\subsection{Research questions}

Because the ranges of exams are generally large (30 to 300), it is infeasible to 
nonparametrically model the joint distribution of the eight test scores.	
Moreover, we have demographic variables such as education, age, and sex. 
It is a nontrivial problem to study the relation between these demographic variables and the test scores. 
Furthermore, the existence of missing test scores among certain individuals further compounds the complexity of the overall analysis.
In this paper, our primary focus centers on addressing the following research questions:
\begin{enumerate}
	\item {\bf Finding a feasible and interpretable model for  multiple discrete outcomes with missing values.}
	As mentioned before, our data contains multiple dependent discrete variables with missing values.
	We need to design a proper model to model the dependency among discrete variables
	and deal with the missing data problem. 
	We also need an estimation and inference procedure that quantifies the uncertainty in the model while properly accounting for the missing data in a statistically principled way.
	
	\item {\bf Discovering latent groups of individuals using neuropsychological tests.}
	Clinicians have developed a set of rules to categorize  individuals 
	into different clinical groups (normal cognition, MCI, and dementia).
	However, this rule is based on clinical judgments and does not involve any neuropsychological information.
	Clustering has been of interest to the dementia research community \citep{alashwal2019application}.
	Thus, an ideal model should be able to cluster individuals into groups using neuropsychological test scores that represent multiple cognitive domains of interest.
	
	\item {\bf Investigating the association between neuropsychological tests and other variables.}
	It is often of a lot of interest to study how neuropsychological tests are associated
	with demographic variables or clinical judgments. 
	For instance, there is a hypothesis that people with a higher education are more resilient to cognitive decline as a form of cognitive reserve \citep{Meng2012EducationAD, ThowMeganE2018Feic}. 
	We also want to use the NACC data to test this hypothesis.

\end{enumerate}

\subsection{Literature review}

We provide some background on the current research in the Alzheimer's disease and statistics methodology literature.

\subsubsection{Dementia-related research}
\label{subsect:dementia-related}

Multiple outcome variables are common in dementia-related research, but there is no clear widespread solution to the modeling problem.  In some previous methods in the Alzheimer's disease literature, multiple test scores are standardized and averaged into a single holistic score for an individual \citep{Boyle2018}.  This turns the multiple outcome problem into a single outcome problem and ignores the dependence structure.  \cite{mjorud2014variables} examined variables associated with multiple quality-of-life related discrete outcomes and performed univariate regression on each outcome.  Dimension reduction is also commonly used to reduce the number of outcome variables \citep{yesavage2016principal, qiu2019cognitive}, but this can lead to loss of interpretability as any downstream analysis is not using the original variables.  Missing data is often ignored completely or when accounted for, single imputation methods or an off-the-shelf imputation approach are typically used \citep{brenowitz2017mixed, qiu2019cognitive}.

Clustering has also become increasingly valuable in the dementia research community \citep{alashwal2019application}.  For example, \cite{Escudero2011} used the $K$-means algorithm to divide subjects into pathologic and non-pathologic groups to study the early detection of Alzheimer's disease.  \cite{tosto2016progression} used $K$-means on a subset of NACC data to identify subgroups within Alzheimer's disease patients to understand the heterogeneity of the disease.  Several papers also described model-based clustering approaches.  For example, \cite{meyer2010diagnosis} clustered biomarker data using a simple two-component Gaussian mixture model for early detection of Alzheimer's disease.  \cite{qiu2019cognitive} used neuropsychological test scores from a smaller data set than ours, imputing missing data using single imputation methods and the \textit{mice} R package \citep{mice}.  Then, they applied principal components analysis before finally, using a Gaussian mixture model for clustering.  In our work, we choose to model the raw test scores to avoid loss of information and preserve more interpretability.  It is crucial to be cautious when employing single imputation methods or off-the-shelf packages, as they can potentially lead to an underestimation of the uncertainty associated with missing data, and the underlying assumptions regarding missing data may remain unclear.

Additionally, the cognitive reserve hypothesis emerged several decades ago when it was observed that after autopsies, some individuals with no dementia symptoms had brains that exhibited advanced Alzheimer's pathology \citep{katzman1989development}.  Researchers hypothesized that some activities in life may provide individuals with a given resilience to cognitive decline, which is known as the cognitive reserve hypothesis.  As such, there is deep interest in discovering the association between various covariates and dementia \citep{zhang1990prevalence, stern2012cognitive, Meng2012EducationAD, ThowMeganE2018Feic}.

\subsubsection{Methodology research}

This paper sits at the intersection of multiple subfields of statistics including latent variable modeling, missing data, and clustering.  The latent class model was first proposed by \cite{Goodman1974} for the purpose of modeling multivariate discrete data.  The basic form of the latent class model is of a mixture of multivariate multinomials.  When the number of levels of the discrete random variable is large, this can lead to a large number of parameters \citep{raftery2019}.  There has been much work studying many aspects of the latent class model including the identifiability \citep{gyllenberg1994, carreira-perpinan2000practical, Allman2009, gu2020} and incorporation of covariates \citep{HUANGGuan-Hua2004Bail, Vermunt2010, ouyang2022identifiability}.  The traditional latent class model has been extended in various forms such as a mixture of Rasch models \citep{Rost1990} and a mixture of item response models \citep{McParland2013}.

The mixture of experts models are closely related to the mixture model framework and the latent class model \citep{Jacobs1991}.  It generalizes the standard mixture model by allowing the parameters to potentially depend on covariates \citep{Gormley2018HandbookOM}.  The increased flexibility of the models is associated with a large increase in number of parameters.  Many estimation procedures have been explored including those based on an EM approach \citep{Jacobs1991, Jordan1993}, an EMM approach \citep{Gormley2010}, or a Bayesian framework \citep{Bishop2002, Fruhwirth-Schnatter2008}.

Clustering has previously been explored in combination with missing data.  \cite{lee2022incomplete} proposed a two stage clustering framework by first clustering multiply imputed data sets individually and then, obtaining a final clustering by clustering the set of cluster centers obtained over all imputed data sets.  In previous research on model-based clustering with missing data, \cite{serafini2020handling} employed the EM algorithm and Monte Carlo methods to estimate Gaussian mixture models in the presence of missing at random data.  \cite{sportisse2023modelbased} addressed model-based clustering under missing-not-at-random assumptions, employing a likelihood-based approach and EM algorithm.  We include further comments on these methods in Appendix \ref{appendix:alternativeEM}.  These methods addressed the clustering problem but did not use covariates or discuss quantifying the uncertainty of the parameters.  Since our method is motivated by the statistical analysis of Alzheimer's disease, handling both of these is critical to our work.

\subsection{Outline}

Our paper is organized as follows.  Section \ref{sect:model-intro} presents a detailed description of our \textit{mixture of binomial product experts} model. The handling of missing data is discussed in Section \ref{sect:missing}, where we describe the fitting of the model under a nonmonotone missing at random assumption. In Section \ref{sect:inference}, we delve into the process of inference, specifically in constructing confidence intervals for the parameters of interest.  We outline how to perform clustering in Section \ref{sect:clustering}. In Section \ref{sect:simulation}, we provide simulations results to examine the consistency and the coverage of our proposed estimator.  In Section \ref{sect:realdata}, we apply our method to an Alzheimer's disease data set, which motivates our model formulation. Finally, Section \ref{sect:discussion} concludes the paper with a discussion of the results. In the Appendix, we include further details on derivations and proofs (Appendices \ref{appendix:EMderivation} and \ref{appendix:proofs}), identification theory (Appendix \ref{appendix:identification}), code implementation (Appendix \ref{appendix:code-implement}), remarks on related EM algorithms (Appendix \ref{appendix:alternativeEM}), simulations (Appendix \ref{sect:simulations}), comments on model selection (Appendix \ref{appendix:modelselect}), comments on clustering (Appendix \ref{appendix:clustering}), and more comments on the real data analysis (Appendix \ref{appendix:more-realdata}). We provide a code implementation of our method in an R package.

\section{Mixture of Binomial Product Experts}

\label{sect:model-intro}

\subsection{A latent class model for neuropsychological test scores}

Suppose we have $n$ IID observations, indexed by $i=1,\ldots,n$.  Let $\bX_i = (X_{i,1}, X_{i,2}, \ldots, X_{i,p})^\top\in\mathcal{X}\subset \mathbb{R}^p$ be the covariates and $\bY_i = (Y_{i,1}, Y_{i,2},\ldots, Y_{i,d})^\top \in \mathcal{Y} \subset \mathbb{Z}^d$ be the outcome variables of interest, where $\mathcal{Y}$ is a bounded discrete set.  In the present study, $\bY_i$ represents the neuropsychological test scores of individual $i$ for a total of $d$ total tests, but this framework can be applied to other problems as well.  In the paper, for notational convenience, we may elect to drop the index $i$ to simplify notation and refer to the generic random variable that is not associated with a specific individual.  The $j$-th test score/outcome variable $Y_j$ belongs to the set $\{0,1,2,\ldots,N_j\}$.  We assume that the covariates are always observed and that only the outcomes may be subject to missingness.

We start with a simpler scenario in which the data is completely observed (no absence of any outcome variables). The more complicated setting with missing outcome variables is addressed in Section \ref{sect:missing}.  When we have multiple discrete outcome variables, the nonparametric approach is not feasible due to the curse of dimensionality. To see this, suppose we have four outcome variables, 
and each outcome variable has ten possible values. The  joint distribution has a dimension of $10^4$, surpassing the order of magnitude of moderately sized samples. This motivates to use a parametric model.

To model the dependency among discrete variables, we consider a mixture model, which is based on
the idea of the latent class model (LCM).
Originally presented in its modern form by \cite{Goodman1974}, the latent class model takes the form of the following mixture of multivariate multinomials distribution
\begin{equation}
	p(y; \theta) = \sum_{k=1}^K w_k p_k(y ; \theta_k)
	= \sum_{k=1}^K w_k \prod_{j=1}^d p_{k,j}(y_{j}; \theta_{k,j}),
\end{equation}
where $p_{k,j}$ denotes a multinomial with $\theta_{k,j}$ being its the parameter and the $w_k$'s are probability weights for each component such that $\sum_k w_k = 1$ and $w_k>0$ for all $k$.  As $p_{k,j}$ is a multinomial distribution, $\theta_{k,j}$ refers to vector of probabilities (one for each level of $Y_j$) that sum to 1.  There is a popular implementation for fitting these models through the \textit{poLCA} package \citep{LinzerLewis2011}.

Let $Z$ denote the component assignment such that $Z = \ell$ if and only if the associated observation is generated from mixture component $\ell$. Namely, $P(Y=y|Z= k;\theta) =  \prod_{j=1}^d p_{k,j}(y_{j}; \theta_{k,j})$.  By construction, the latent class model implicitly assumes for $j\neq j'$, the variables $Y_j$ and $Y_{j'}$ are conditionally independent given the latent class $Z$.  This is known as the local independence assumption and is commonly used in the literature to model dependence between discrete random variables.  While they are conditionally independent, the variables $Y_j$ and $Y_{j'}$ are allowed to depend marginally for $j\neq j'$.  The local independence assumption makes it convenient to calculate the conditional distributions, which are used in Section \ref{sect:missing} when we discuss imputation.  Some effort has been made to relax the local independence assumption such as mixtures of log-linear models \citep{Bock1994} and hiearchical latent class models \citep{Zhang2004, Chen2012}, but these are have not been as widely adopted due to the increase in the number of parameters.

Throughout this paper, we will refer to the data set containing tuples of the form $(\bX,\bY)$ as the \textit{latent incomplete} (LI) data and the data set containing tuples of the form $(\bX,\bY,Z)$ as the \textit{latent complete} (LC) data.  In the setting of no missing data, the latent incomplete data is what we typically have access to, but this is often insufficient for direct model fitting due to the unobserved latent variable $Z$.  We will show that one can estimate the model parameters using an EM algorithm with the latent incomplete data.  We use this nomenclature to avoid confusion in Section \ref{sect:missing} when we encounter the traditional type of missingness with the outcome variables.

There has been previous work to incorporate covariates in the latent class model \citep{HUANGGuan-Hua2004Bail, Vermunt2010}.  This is related to
the \textit{mixture of experts} models that arose from the machine learning literature \citep{Jacobs1991, Jordan1993, Yuksel2012}.  These models generalize the classical mixture model formulation by allowing the model parameters $w_k$ and $\theta_k$ to depend on covariates.  The \textit{simple mixture of experts} model \citep{Gormley2018HandbookOM} takes the following form
\begin{equation*}
	p(y | x; \btheta) = \sum_{k=1}^K w_k(x;\beta_k) p_k(y ; \theta_k),
\end{equation*}
where $\beta_k$ is introduced to allow the weights to depend on covariates.  As written, the covariates only adjust the weights placed on each mixture component but do not affect the parameters of the individual component distribution.

To apply these ideas to the neuropsychological data, we combine the mixture of experts and LCM, leading to 
the following mixture of binomial product experts
\begin{align}
	p(y|x; \bbeta, \btheta) = \sum_{k=1}^K w_k(x;\beta_k) \underbrace{\prod_{j=1}^d \underbrace{\binom{N_j}{y_j} (\theta_{k,j})^{y_j} (1-\theta_{k,j})^{N_j-y_j}}_{p_{k,j}(y_j;\theta_{k,j})}}_{p_k(y;\theta_k)},
\end{align}
where $w_k(x;\beta_k) = \exp(\beta_k^\top \cdot(1,x)) / (\sum_{k'\in[K]} \exp(\beta_{k'}^\top \cdot(1,x)))$ is from the multiple-class logistic regression model.  Like in the classical LCM, we utilize the local independence assumption as each mixture decomposes as a product of binomials.  However, there is a critical difference between our model and the classical LCM.  We use the binomial distribution instead of the multinomial distribution, and this leads to large reduction in the number of parameters.  Since test scores are ordered variables, we expect the reduction to a binomial distribution to be reasonable as scores further away from the mean score can be less likely.  Similar to the simple mixture of experts model, each mixture weight depends on covariates.

Here we will treat $k=1$ as the reference group, so $\beta_1$ is the vector of all $0$ by assumption and does not need to be estimated.  For $k \in \{2,3,\ldots,K\}$, $\beta_k \in \mathbb{R}^{p+1}$, and for $k\in[K]$, $\theta_k \in \mathbb{R}^d$, we assume that the $w_k$  utilize the logistic function, but it is possible to use other link functions such as the probit and other monotonic functions \citep{ouyang2022identifiability}.  In our work, the outcome variables are neuropsychological test scores, and we treat the latent groups as subgroups of the population of varying cognitive ability.  The five levels of the CDR score offer a natural choice for the number of mixture components.

In our model construction, we view each mixture component as representing a subgroup of the population, where the subgroups have varying cognitive ability that is summarized by $\btheta$.  Therefore, we interpret the $\btheta$'s as attributes of population-level groups and do not allow them to depend on covariates.  Since the weights $w_k(x;\beta_k)$ depend on covariates, every individual has their own weights associated with each population-level group.  These covariates can be interpreted as a way to construct weights on an individual level for each of these population-level subgroups.  A side product of this assumption is that this reduces model complexity. If the $\btheta$'s depend on covariates, then the number of parameters increases from $O(K(d + p))$ to $O(Kdp)$, which is a dramatic change even for moderately sized $d$ and $p$ and thus another reason we do not allow this as part of our model.  Because the outcome variables are test scores, we can also interpret the quantity $N_j \cdot \theta_{k,j}$ as the mean of the $j$-th test score of the $k$-th latent group, and this is useful as a summary statistic for a given latent class.

Binomial product distributions have previously arisen in the literature.  One such application was for modeling test scores for spatial tasks in child development in \cite{Hoben1993}, but they considered a fairly limited setting with no covariates, two components, and two outcome variables.  Binomial product and Poisson product distributions have also been used in the statistical ecology literature to model species abundance \citep{Kery2005, Haines2016, Brintz2018}.  These settings differ from ours because they often treat the total number counts $N_j$ as a random quantity.  In item response theory (IRT), the Rasch model posits that the probability of answering each question correctly depends on its difficulty and the individual's ability, assuming a Bernoulli distribution for each question \citep{rasch1960probabilistic}. However, in our data, only total test scores are available.  Thus, assuming uniform difficulty across test questions, we model the scores using Binomial distributions.  To our knowledge, our paper is the first time binomial products has been used to analyze neuropsychological test score data while incorporating covariates.

Note that for the $k$-th component, $\beta_k \in\mathbb{R}^{p+1}$ and $\theta_k \in\mathbb{R}^d$, there are $(K-1)(p+1)$ parameters associated with the covariates $\bbeta$ and $Kd$ parameters associated with $\btheta$ and a total of $K(p+d+1)-p-1$ parameters for the whole model.  Thus, for fixed covariate and outcome dimensions, the number of parameters grows linearly in $K$.

\subsection{Model fitting}

We now describe our procedure of fitting the model. 
An intuitive approach is to estimate the parameters via the maximum likelihood approach.
However, computing the maximum likelihood estimator (MLE) is a nontrivial task because the latent incomplete log-likelihood is not concave (Lemma \ref{lemma:nonconcavity}). As such, solving the first order conditions is no longer sufficient for determining the global maximizer.  

Due to Lemma \ref{lemma:nonconcavity} (stated formally in Appendix \ref{appendix:lemmaproof}), maximizing the latent incomplete log-likelihood function directly is not straightforward. 
Here is an interesting insight from the mixture model literature: if we had access to the latent class label, then the maximum likelihood estimator would be easy to construct.  This insight motivates us to develop
an EM algorithm \citep{Dempster1977} by augmenting the observed data with the true component label $Z$  to calculate the MLE for $(\bbeta, \btheta)$.  For each $i=1,\ldots,n$, let $Z_i = \ell$ if the $i$-th observation comes from mixture component $\ell$.  Such data augmentation allows one to bypass the problem of taking the logarithm of a sum.  To see this, let $\bX_{1:n} = (\bX_1, \ldots, \bX_n)$ and $\bY_{1:n} = (\bY_1, \ldots, \bY_n)$ be the covariates and outcome variables of all $n$ individuals, respectively.  The latent complete (LC) log-likelihood function writes as
\begin{align}
	\ell_{\text{LC},n}&(\beta, \theta ; \bX_{1:n}, \bY_{1:n}, \bZ_{1:n}) = \log \left( \prod_{i=1}^n \prod_{k=1}^K \left(w_k(\bX_i; \beta_k) p_k(\bY_i ; \theta_k)\right)^{1(Z_i=k)}\right) \nonumber \\
	&= \sum_{i=1}^n\sum_{k=1}^K 1(Z_i=k) 
	\log w_k(\bX_i;\beta_k) + \sum_{i=1}^n\sum_{k=1}^K 1(Z_i=k)\log p_k(\bY_i;\theta_k). \label{eq:lc-loglike}
\end{align}

In the EM algorithm, we start from an initial guess $(\hat{\bbeta}^{(0)}, \hat{\btheta}^{(0)})$ and iterate between an expectation step (E-step) and a maximization step (M-step) to update our guess until convergence.  Let $t$ denote the $t$-th iteration of the EM algorithm.  In the E-step, we compute the expected value of the complete log-likelihood $\ell_{\text{LC}}(\bbeta, \btheta ; \bX,\bY,\bZ)$ conditional on the observed data $\bX$ and $\bY$.  The expected value of the latent complete log-likelihood given the observed data forms the $Q$ function in the EM algorithm
\begin{equation}
	Q_{\text{LC}}(\bbeta, \btheta | \bX, \bY; \hat{\bbeta}^{(t)}, \hat{\btheta}^{(t)}) := \E[\ell_{\text{LC}}(\bbeta, \btheta ; \bX,\bY,Z) | \bX, \bY; \hat{\bbeta}^{(t)}, \hat{\btheta}^{(t)}] \nonumber.
\end{equation}
In practice, we do not have access to the true expectation, so we work with the sample analogue $Q_{\text{LC},n}$.  The sample analogue can be expressed as
\begin{align}
	&Q_{\text{LC}, n}(\bbeta, \btheta | \bX_{1:n}, \bY_{1:n}; \hat{\bbeta}^{(t)}, \hat{\btheta}^{(t)}) \nonumber \\
	&:= \frac{1}{n}\sum_{i=1}^n\sum_{k=1}^K \hat{\pi}_k^{(t)}(\bX_i, \bY_i)
	\log w_k(\bX_i;\beta_k) + \frac{1}{n}\sum_{i=1}^n\sum_{k=1}^K \hat{\pi}_k^{(t)}(\bX_i, \bY_i) \log p_k(\bY_i;\theta_k),
\end{align}
where
\begin{equation}
	\hat{\pi}_k^{(t)}(\bX_i, \bY_i) := P(Z=k|\bX_i,\bY_i; \hat{\bbeta}^{(t)}, \hat{\btheta}^{(t)}) = \frac{w_k(\bX_i;\hat{\beta}_k^{(t)}) p_k(\bY_i; \hat{\theta}_k^{(t)})}{\sum_{k'=1}^K w_{k'}(\bX_i;\hat{\beta}_{k'}^{(t)}) p_{k'}(\bY_i; \hat{\theta}_{k'}^{(t)})} \label{eq:weights}
\end{equation}
is the weight of observation $(\bX_i, \bY_i)$ for the $k$-th mixture component.

Since the latent-complete log-likelihood function \eqref{eq:lc-loglike} decomposes as the sum of a function of $\bbeta$ and a function of $\btheta$, the maximization step of the EM algorithm is separable.  Note that since $w_k$ and $p_k$ are logistic regression and binomial product models, respectively, one can apply standard off-the-shelf model fitting procedures after reweighting each observation $i$ by $\hat{\pi}_k^{(t)}(\bX_i,\bY_i)$.  In the maximization step, the new estimate for $\bm{\theta}$ is updated as follows for each $k\in[K]$ and $j \in [d]$,
\begin{equation}
	\hat{\theta}^{(t+1)}_{k, j} := \dfrac{\sum_{i=1}^n \frac{Y_{i,j}}{N_j} \hat{\pi}_k^{(t)}(\bX_i, \bY_i)}{\sum_{i=1}^n \hat{\pi}_k^{(t)}(\bX_i, \bY_i)}, \label{eq:thetaUpd}
\end{equation}
which is simply the standard MLE formed by the sample proportion, reweighted by $\hat{\pi}_k^{(t)}(\bX_i,\bY_i)$.  We update our estimate of $\beta_k$ by fitting a multi-class logistic regression by reweighting each corresponding observation by $\hat{\pi}_k^{(t)}(\bX_i, \bY_i)$.  This is equivalent to regressing the variable $W^{(t)}$ on the covariates $X$, where $W_i^{(t)} = (\hat{\pi}_1^{(t)}(\bX_i, \bY_i), \ldots, \hat{\pi}_K^{(t)}(\bX_i, \bY_i)) \in S^{K-1}$ for all $i$, where $S^{K-1}$ is the $(K-1)$-th probability simplex.  Note that this is not the standard logistic regression because the outcome variable belongs to a probability simplex.  This logistic regression can still be fit using gradient descent in standard R packages.

\begin{algorithm}
	{\bf Require:} $\{(\bX_i, \bY_i)\}_{i=1}^n$, $(\hat{\bbeta}^{(0)}, \hat{\btheta}^{(0)})$
	\begin{algorithmic}[1]
		\State{Start with an initial guess $(\hat{\bbeta}^{(0)}, \hat{\btheta}^{(0)})$.}
		\State{Initialize the iteration step $t=0$.}
		\Repeat
		\State{Compute the weights $\hat{\pi}_k^{(t)}(\bX_i, \bY_i)$ for all observations $i$ and mixture components $k$ using \eqref{eq:weights}.}
		\State{Form the probability vector $W_i^{(t)} = (\hat{\pi}_1^{(t)}(\bX_i, \bY_i), \ldots, \hat{\pi}_K^{(t)}(\bX_i, \bY_i))$ for each $i$.}
		\State{Update $\hat{\btheta}^{(t+1)}$ using \eqref{eq:thetaUpd}.}
		\State{Update $\hat{\bbeta}^{(t+1)}$ by fitting a logistic regression model of $W^{(t)}$ on $\bX$.}
		\State{Update $t \leftarrow t+1$.}
		\Until{Convergence}
		\State{Let $\hat{\bbeta} = \lim_{t\to\infty} \hat{\bbeta}^{(t)}$ and $\hat{\btheta} = \lim_{t\to\infty} \hat{\btheta}^{(t)}$.}
		\State{\textbf{return} $(\hat{\bbeta}, \hat{\btheta})$}
	\end{algorithmic}
	\caption{Fitting the mixture model via EM} \label{alg:latentEM}
\end{algorithm}

Algorithm \ref{alg:latentEM} describes the process of fitting the mixture of binomial products experts model in the presence of completely observed covariate and outcome data.  For a single initialization point, the EM algorithm is not guaranteed to converge to the global maximizer.  In practice, we run Algorithm \ref{alg:latentEM} many times with multiple random initial guesses to explore the parameter space sufficiently and choose the parameter estimate from all of them that maximizes the log-likelihood.

\subsection{Identifiability}

We now make some brief comments on the identifiability of our model.  Identifiability means that each parameter value corresponds to a unique probability distribution.  That is, the mapping from the parameter space to the space of probability distributions is one-to-one. If identifiability does not hold, estimation becomes problematic because a unique MLE may not exist.  However, this notion of identifiability may be too strong for many practical purposes.  For example, even the common problem of label swapping violates this identifiability definition, so researchers often consider the idea of identifiability up to permutation of the parameters.  We now consider the notion of \textit{generic identifiability}, which relaxes the definition of identifiability even further.  Generic identifiability means identifiability holds almost everywhere in the parameter space \citep{Allman2009}.  More formally, this means that the mapping from the parameter space to the space of probability distributions may fail to be one-to-one only on a set of Lebesgue measure zero.  Generic identifiability (holding almost everywhere) can be viewed as an intermediate assumption between two common assumptions: global identifiability (holding everywhere) and local identifiability (holding in a neighborhood of the true parameter).  In a practical sense, this means that any such model fit on a given data set is unlikely to be unidentifiable, and we consider this notion to be sufficient for our applied data setting.
	
	\begin{proposition}[Sufficient conditions for generic identifiability]
		\label{prop:identif}
		Suppose the following conditions hold.
		\begin{enumerate}[label=A\arabic*]
			\item Each mixture is distinct such that $\theta_k \neq \theta_{k'}$ when $k \neq k'$. \label{assump:distinctmix}
			\item The number of outcome variables $d$ and the number of mixtures $K$ satisfies the bound $d \geq 2\lceil \log_{1+\min_j N_j} K \rceil + 1$. \label{assump:inequality}
			\item The design matrix is full-rank and $n> p$.\label{assump:identified-logistic} 
		\end{enumerate}
		Then, the mixture of binomial product experts is generically identifiable up to permutation of the parameters.
	\end{proposition}
	
	We note that the sufficient conditions outlined in Proposition \ref{prop:identif} are fairly mild.  Intuitively, one would not want the number of parameters to be too large when the data dimension $d$ is small. The bound in assumption \ref{assump:inequality} ensures that the number of mixture components $K$ remains appropriately controlled relative to $d$ and $N_j$, as a large $K$ increases model complexity. When there are at least $d=3$ outcome variables of moderate range, $K$ can still be fairly large.  Assumption \ref{assump:identified-logistic} is to ensure identifiability of the logistic regression parameters. We provide further comments on the assumptions and the proof of this proposition in Appendix \ref{appendix:identification}.  These assumptions are fairly mild, and we show that they can be met using relatively straightforward generating processes through our simulations in Section \ref{sect:simulation}. 

\section{Missingness in the Outcome Variables}

\label{sect:missing}

Another challenge in the NACC's neuropsychological data is the missingness in the outcomes.
To properly describe the missing data problem, we first introduce some additional notation.  Let $\bR_i = (R_{i,1}, R_{i,2}, \ldots, R_{i,d})^\top\in\mathcal{R}\subseteq\{0,1\}^d$ be the random binary vector that denotes the missing pattern associated with individual $i$.  Element $R_{i,j}$ of the binary vector $\bR_i$ is $1$ if and only if $Y_{i,j}$ is observed.  For a given pattern $r \in\mathcal{R}$, let $\bY_{i,r} = (Y_i,j: r_{i,j}=1)$ denote the observed variables and $\bY_{i,\bar{r}} = (Y_{i,j}: r_{i,j}=0)$ denote the missing variables.  For example, when $d=4$ and $r=1001$, then $\bY_{i,r} = (Y_{i,1},Y_{i,4})$ and $\bY_{i,\bar{r}} = (Y_{i,2}, Y_{i,3})$.  Similarly, for observation $i$ with random missing pattern $\bR_i$, denote $\bY_{i,\bR_i}$ and $\bY_{i,\bar{\bR}_i}$ as the observed and missing outcome variables for the $i$th observation, respectively.  We place no restrictions on the set of possible patterns $\mathcal{R}$, allowing the pattern to be nonmonotone.  Since $|\mathcal{R}| \leq 2^d$, missingness can easily become an exponentially complex problem in the dimension of the outcome variables.

In contrast to the previous section, we will use the term \textit{fully complete} (FC) to refer to a data set containing IID tuples of the form $(\bX,\bY,Z,\bR)$.  We call a data set containing IID observations of the form $(\bX,\bY_{\bR},\bR)$ as the \textit{observed} data.  Note that the observed data now has two kinds of incompleteness/missingness: incompleteness in the latent class variable (due to mixture models) and missingness in the outcomes.  As the case with the latent complete data, we also do not have access to the fully complete data.  We collect all of the observed outcome variables in one tuple $\bY_{1:n, \bR_{1:n}} = (\bY_{1,\bR_1}, \bY_{2,\bR_2}, \ldots, \bY_{n,\bR_n})$.

\subsection{Missing at random and an imputation strategy}
\cite{rubin1976inference} outlined three types of missingness mechanisms: missing completely at random (MCAR), missing at random (MAR), and missing not at random (MNAR).  Missing completely at random assumes that the missingness of the variable is independent of all variables in the data.  Missing at random assumes that the missingness of a variable can only depend on observed variables.  Missing not at random assumes that the missingness of a variable can depend on the value of the variable subject to missingness.

In the  MCAR data, model fitting is straightforward because one can run Algorithm \ref{alg:latentEM} on the complete cases, and the estimates of the parameters will be consistent.  
However, MCAR assumes that the missingness is irrelevant to observed outcomes, which is clearly false
in the NACC data as a possible
reason to miss some test scores is that the individual was too sick to finish 
the test. 
Therefore, we consider the MAR assumption.	Formally, the definition of missing at random (MAR) is as follows.

\begin{definition}[Missing at random]
	The outcome variables $\bY$ are missing at random (MAR) if the probability of missingness is dependent only on the variables that are observed for the given pattern.  This assumption writes as
	\begin{equation}
		P(\bR=r | \bX, \bY) \stackrel{\text{MAR}}{=} P(\bR=r | \bX, \bY_r) \label{eq:mar}
	\end{equation}
	for all $r\in\mathcal{R}$.
\end{definition}

Notice that the left-hand side of \eqref{eq:mar} represents the selection probability $P(\bR=r|\bX,\bY)$. This quantity is strictly unidentifiable because it depends on unobserved data, specifically $\bY_{\bar{r}}$.  Thus, it cannot be estimated even given infinite observed data.  The missing at random assumption makes $P(\bR=r|\bX,\bY)$ identifiable because it equates it to a probability $P(\bR=r|\bX,\bY_r)$ that can be estimated from the observed data; the variables under the conditioning are all of the variables strictly observed under pattern $r$.  The definition of missing at random implies that the probability of a given missing pattern does not depend on variables that are unobserved under that pattern.  A key advantage of this assumption is that we avoid the challenge of modeling the joint distribution between $\bY$ and $\bR$ since the missingness mechanism $P(\bR=r|\bX,\bY)$ does not need to be modeled directly; this is known as the \textit{ignorability property} (see the discussion later).  Thus, we do not have to deal with making potentially unreasonable modeling assumptions on either the selection model $P(\bR=r|x,y)$ or the extrapolation distribution $p(y_{\bar{r}} | y_r, x, \bR=r)$.

The MAR assumption is untestable because its validity relies strictly on data that is unobserved, and so it cannot be rejected by the observed data.  Our primary reason for selecting this assumption is for modeling.  This may lead to easier interpretability for scientists and practioners since we fit a global model $p(y|x)$ across all missing patterns rather than a local model $p(y|x, R = r)$ for every pattern $r$. In the Alzheimer's disease literature, having more statistically sound work can be meaningful because of limitations in existing work that we have described previously in Section \ref{subsect:dementia-related}.  Since we believe that it may be plausible because these tests are correlated due to underlying cognitive ability, we can use it as a starting point. This can be viewed as a baseline before proceeding with more complex MNAR assumptions. We recognize that MNAR may be more reasonable since missing test scores may be attributed to sickness.  However, this is in itself a very open research question because MNAR is a broad class of assumptions, and performing mixture modeling with MNAR is not straightforward; we leave this for future work.

An additional feature of the MAR assumption is that 
this assumption offers a simple approach to impute the data, which 
makes the computation of the MLE a lot easier. 
Before describing the procedure of updating model parameters,
we first introduce a multiple imputation procedure in Algorithm \ref{alg:MI} that can be used to fill in the missing data.  In the algorithm, for notational convenience, when we write a binary vector in the summation or product, this means we iterate over all indices whose elements are nonzero.  For instance, when we write ``For $j$ in $(1,0,0,1)$,'' this is equivalent to ``For $j=1,4$.''  Therefore, if $\bR_i = (1,0,0,1)$, then ``For $j$ in $\bR_i$'' corresponds to ``For $j=1,4$'' as well.

\begin{algorithm}
	{\bf Require:} $\{(\bX_i, \bY_{i,\bR_i}, \bR_i)\}_{i=1}^n$, $(\hat{\bbeta}, \hat{\btheta})$, $M$ (the number of imputed data sets)
	\begin{algorithmic}[1]
		\For{$i=1,\ldots,n$}
		\If{$\bR_i \neq 1_d$:}
		\For{$k=1,\ldots,K$}
		\State{Form the weights $\omega_{i,k} = w_k(\bX_i;\hat{\beta}_k) p_{k,\bR_i}(\bY_{i,\bR_i};\hat{\theta}_k) / \sum_{k'} w_{ k'}(\bX_i;\hat{\beta}_{ k'}) p_{ k',\bR_i}(\bY_{i,\bR_i};\hat{\theta}_{ k'})$}
		\EndFor
		\EndIf
		\EndFor
		\For{$m = 1,\ldots,M$}
		\For{$i=1,\ldots,n$}
		\If{$\bR_i \neq 1_d$}
		\State{Sample $Z_i^{(m)} \sim \text{Categorical}(K, p=(\omega_{i,1},\ldots,\omega_{i,K}))$}
		\For {$j\in \bar{\bR}_i$}
		\State{Sample $\tilde{Y}_{i,j}^{(m)}$ from $\text{Binomial}\left(N_j, \hat{\theta}_{Z_i^{(m)},j}\right)$.}
		\EndFor
		\State{Assign $\tilde{\bY}_i^{(m)} := (\bY_{i,\bR_i}, \tilde{\bY}_{i,\bar{\bR}_i}^{(m)})$.}
		\EndIf
		
		\EndFor
		
		\State{Let $\tilde{D}_m = \{(\bX_i, \tilde{\bY}_i^{(m)}, \bR_i)\}_{i=1}^n$.}
		
		\EndFor
		
		\State{\textbf{return} $\{\tilde{D}_m\}_{m=1}^M$}
	\end{algorithmic}
	\caption{Multiple imputation}
	\label{alg:MI}
\end{algorithm}

As stated, Algorithm \ref{alg:MI} describes how to impute the outcome variables, assuming that a good estimate $(\hat{\bbeta}, \hat{\btheta})$ is available.  We discuss how to actually obtain such an estimate in the Section \ref{subsect:modelfitting-missing}.  The multiple imputation algorithm exploits the fact that the conditional distribution $p(y_{\bar{r}}|y_r,x)$ for any $r\in\mathcal{R}$ remains a mixture model.  There are two steps to performing multiple imputation: 1) for every observation $i$ with missing observations, we compute the weights of each component of the mixture distribution $p(y_{\bar{\bR}_i}|\bY_{i,\bR_i},\bX_i)$, and 2) sample $M$ times from the distribution $p(y_{\bar{\bR}_i}|\bY_{i,\bR_i},\bX_i)$ for each observation $i$ to form $M$ completed data sets.  The derivation of this procedure is provided in Appendix \ref{appendix:multiple-imp}.


\subsection{Model fitting under a missing at random assumption}

\label{subsect:modelfitting-missing}

For each $r$, we assume that the selection probability $P(\bR=r|\bX,\bY_r; \gamma_r)$ belongs to a parametric family, indexed by $\gamma_r$.  We collect these parameters into $\bgamma = (\gamma_r)_{r\in\mathcal{R}}$.  
For simplicity, we write the log-likelihood in terms of the probability model without the $n$ samples.  
Under the MAR assumption, we can write the observed log-likelihood as
\begin{align*}
	\ell_{\text{obs}}(\bgamma, \bbeta, \btheta ; x, y_r, r) &= \log P(\bR=r | x,y_r; \gamma) + \log p(y_r | x; \bbeta, \btheta) \\
	&= \ell_{\text{obs}}^{(1)}(\bgamma ; x, y_r, r) + \ell_{\text{obs}}^{(2)}(\bbeta, \btheta ; y_r, x).
\end{align*}

The missingness mechanism is said to be ignorable because estimation of $(\bbeta, \btheta)$ is separated from the estimation of $\bgamma$.  Model fitting under a missing at random assumption typically occurs using an EM algorithm approach.  We can proceed with estimating $(\bbeta, \btheta)$ by conditioning on the observed data $(\bX,\bY_\bR)$ and using another EM algorithm approach.  The population $Q_{\text{FC}, r}$ function writes as follows
\begin{align}
	Q_{\text{FC}, r}(\bbeta, \btheta | \bX, \bY_r; \hat{\bbeta}^{(t)}, \hat{\btheta}^{(t)}) &:= \E[\ell_{\text{LI}}(\bbeta, \btheta ; \bX, \bY) | \bX, \bY_r ; \hat{\bbeta}^{(t)}, \hat{\btheta}^{(t)}] \nonumber \\
	&= \sum_{y_{\bar{r}}} \ell_{\text{LI}}(\bbeta, \btheta ; \bX, y_{\bar{r}}, \bY_r) p (y_{\bar{r}} | \bY_r, \bX; \hat{\bbeta}^{(t)}, \hat{\btheta}^{(t)}) \label{eq:Q-FCr}
\end{align}
for every $r \in\mathcal{R}$.  One major observation is that the conditional expectation relies on being able to fit estimate $p(y_{\bar{r}} | y_r, x; \bbeta, \btheta)$.  We can leverage the consequences of the local independence assumption to obtain an easily computable form of the conditional distribution.  Unfortunately, since the latent-incomplete log-likelihood function $\ell_{\text{LI}}$ is not linear in the outcome variables, we are unable to reduce \eqref{eq:Q-FCr} to a more simple form.  However, we can approximate this expectation stochastically by imputing the missing data $\bY_{\bar{r}}$ for every missing pattern $r$ using the distribution $p(y_{\bar{r}} | y_r, x ; \hat{\bbeta}^{(t)}, \hat{\btheta}^{(t)})$.

To overcome estimating the conditional expectation of factorial terms, we approximate \eqref{eq:Q-FCr} stochastically with a Monte Carlo procedure.  For large enough $M$, we expect that
\begin{align*}
	Q_{\text{FC}, r}(\bbeta, \btheta | \bX, \bY_r; \hat{\bbeta}^{(t)}, \hat{\btheta}^{(t)}) &\approx Q_{\text{FC},r,n}^{(M)}(\bbeta, \btheta | \bX_{1:n}, \bY_{1:n,\bR_{1:n}}; \hat{\bbeta}^{(t)}, \hat{\btheta}^{(t)}) \\ &:= \frac{1}{Mn}\sum_{m=1}^M \sum_{i=1}^n \log \left( \sum_{k=1}^K w_k(\bX_i; \beta_k) p_k(\tilde{\bY}_{i}^{(m;t)} ; \theta_k)\right) \cdot 1(\bR_i = r),
\end{align*}

where $\tilde{\bY}_i^{(m;t)} = (\bY_{i,\bR_i}, \tilde{\bY}_{i,\bar{\bR}_i}^{(m;t)})$ denotes the $m$th imputed data for the $i$th observation on the iteration step $t$ given the observed variables.  The choice of number of imputations $M$ is important, and in practice, we use $M=20$ to balance both computational time and good estimation performance.  We discuss this in more detail with our simulations in Section \ref{sect:simulation} and Appendix \ref{appendix:add-sim}.

In the maximization step, we compute the MLE of $(\bbeta, \btheta)$ using the completed data after the multiple imputation.  Since we have completed data, the MLE can be found using the original EM algorithm, outlined in Algorithm \ref{alg:latentEM}.  We summarize the procedure in the following algorithm.  Throughout this paper, we will use the notation $\hat{\bbeta}$ and $\hat{\btheta}$ (without superscripts relating to $t$) to denote the point estimate obtained by Algorithm \ref{alg:marEM} after convergence is achieved.  In practice, we assume that convergence is achieved if the $L_2$ norm of the difference between the current and old parameter estimates falls within a prespecified tolerance level (we choose $\epsilon=10^{-4}$).

\begin{algorithm}
	{\bf Require:} $\{(\bX_i, \bY_{i,\bR_i}, \bR_i)\}_{i=1}^n$, $(\hat{\bbeta}^{(0)}, \hat{\btheta}^{(0)}), M$
	\begin{algorithmic}[1]
		\State{Start with an initial guess $(\hat{\bbeta}^{(0)}, \hat{\btheta}^{(0)})$.}
		\State{Initialize the iteration step $t=0$.}
		\Repeat
		\State{Let $\tilde{D}_t = \left\{\left(\bX_i, \tilde{\bY}_i^{(m;t)}\right)\right\}_{i=1,\ldots,n;\ m=1,\ldots,M}$ be the stacked imputed data sets obtained from Algorithm \ref{alg:MI}.}
		\State{Obtain $(\hat{\bbeta}^{(t+1)}, \hat{\btheta}^{(t+1)})$ using Algorithm \ref{alg:latentEM} on the imputed data set $\tilde{D}_t$ and starting point $(\hat{\bbeta}^{(t)}, \hat{\btheta}^{(t)})$.}
		\State{Update $t\leftarrow t+1$.}
		\Until{Convergence}
		\State{Let $\hat{\bbeta} = \lim_{t\to\infty} \hat{\bbeta}^{(t)}$ and $\hat{\btheta} = \lim_{t\to\infty} \hat{\btheta}^{(t)}$.}
		\State{\textbf{return} $(\hat{\bbeta}, \hat{\btheta})$}
	\end{algorithmic}
	\caption{Fitting the mixture model with MAR outcomes via a nested EM procedure} \label{alg:marEM}
\end{algorithm}

Algorithm \ref{alg:marEM} is a nested EM procedure because we have two types of missingness: missingness in the outcome variables and missingness in the latent class labels.  We embed an EM algorithm for the latent 
	class model fitting (this serves as the M-step in the outer EM algorithm) within an outer EM algorithm for handling the missing at random data.  Because we do not obtain a closed-form expression for the conditional expectations in the E-step but rather perform a stochastic approximation, Algorithm \ref{alg:marEM} is a Monte Carlo EM algorithm \citep{Levine2001}.  We use the notation $\tilde{\bY}_i^{(m;t)}$ to denote the $m$th imputed outcome variables using the model parameterized by $(\hat{\bbeta}^{(t)}, \hat{\btheta}^{(t)})$.  When the parameter is understood, we omit the $t$ index.  
	Note that when there is no missingness in the outcome variables, Algorithm \ref{alg:marEM} reduces to Algorithm \ref{alg:latentEM} because the multiple imputation step is bypassed.  Again, in practice, we run Algorithm \ref{alg:marEM} with multiple random initial points to ensure we explore the parameter space and converge to the global maximizer.  We pick the point estimate that maximizes the observed log-likelihood, defined as
\begin{align}
	\ell_{\text{obs},n} := \ell_{\text{obs},n}(\bbeta, \btheta ; \bX_{1:n}, \bY_{1:n, \bR_{1:n}}) := \sum_{i=1}^n \log \left( \sum_{k=1}^K w_k(\bX_i) p_{k,\bR_i}(\bY_{i,\bR_i})\right), \label{eq:obs-loglike}
\end{align}
where $p_{k,r}(y_r) = \prod_{j\in r} p_{k,j}(y_k)$ for any $r\in\mathcal{R}$.  Note that when there is no missing data, this reduces to the latent incomplete log-likelihood function $\ell_{\text{LI},n}$.  Since we are working under a parametric model and estimating parameters using maximum likelihood, the estimators converge at $\sqrt{n}$-rate and are asymptotically efficient.  We provide theoretical justification for these methods and discussion of the asymptotic behavior in Appendix \ref{appendix:EMderivation}.

\section{Inference}

\label{sect:inference}

In this section, we discuss how to perform inference on the parameters $\btheta$ and $\bbeta$.  There has been some previous work in obtaining asymptotic variance estimators for multiple imputations estimators \citep{WangRobins1998, RobinsWang2000}, but this can be analytically challenging in our setting.  The bootstrap is a widely adopted procedure for estimating the sampling distribution of an estimator and obtaining asymptotically valid confidence intervals \citep{efron1979bootstrap} and has been previously used with the EM algorithm \citep{CeleuxSEM1987}.  \cite{OHagan2019} have also explored the use of the bootstrap with a Gaussian mixture model.  We provide theoretical justification for the bootstrap in our setting in Appendix \ref{appendix:bootstrap}.  

Bootstrapping works by resampling from the original data set (which is equivalent to sampling from empirical distribution), and with large enough sample size and under sufficient regularity conditions, this mimics drawing samples from the true generating distribution.  We bootstrap the data $B$ times (where $B$ is sufficiently large) including the missing data, and the estimation procedure is run on each of the $B$ bootstrapped data sets.

Since our EM algorithm requires many initialization points in practice to properly explore the parameter space, there is a question of how to initialize the bootstrap procedure.  We follow the recommendation outlined in \cite{Chen2022StatisticalOptima} to initialize the bootstrap at the same initial point on every iteration, using the point estimator $(\hat{\bbeta}, \hat{\btheta})$ returned from Algorithm \ref{alg:marEM}.  The primary goal of the bootstrap is to measure the stochastic variation of an estimator around the parameter of interest.  If we perform the bootstrap with random initialization, we will capture additional uncertainty that arises from estimating different local modes of the log-likelihood function.  Initializing at the same point also avoids the label switching identifiability problem, which occurs when the probability distribution remains identical after some parameters are permuted.  Additionally, this saves on computational time because we are also not performing many random initializations.  Our bootstrap procedure is summarized in the Algorithm \ref{alg:bootstrap}.

\begin{algorithm}
{\bf Require:} $\{(\bX_i, \bY_{i,\bR_i}, \bR_i)\}_{i=1}^n$, $(\hat{\bbeta}, \hat{\btheta})$, $B$ (a large number, say 1,000)
\begin{algorithmic}[1]
	\For{$b\in1,\ldots,B$}
	\State{Sample $n$ draws uniformly with replacement from $\{1, 2, \ldots, n\}$.  Put these into index set $I_b$.}
	\State{Set the $b$th bootstrapped data set $D_b^* = \{\bX_i, \bY_{i,\bR_i}, \bR_i\}_{i\in I_b}$.}
	\State{Obtain $(\bbeta^{*(b)}, \btheta^{*(b)})$ using Algorithm \ref{alg:marEM} on data set $D^*_b$ initialized with $(\hat{\bbeta}, \hat{\btheta})$.}
	\EndFor
	
	\State{\textbf{return} $\{\bbeta^{*(b)}\}_{b=1}^B, \{\btheta^{*(b)}\}_{b=1}^B$}
\end{algorithmic}
\caption{Bootstrap procedure for obtaining confidence intervals}
\label{alg:bootstrap}
\end{algorithm}

We now describe how to construct confidence intervals for a given parameter, using $\theta_{1,1}$ as an illustrating example.  We first obtain a point estimate $(\hat{\bbeta}, \hat{\btheta})$ using Algorithm \ref{alg:marEM}.  Then, we run Algorithm \ref{alg:bootstrap} using $(\hat{\bbeta}, \hat{\btheta})$ as the initialization point for a large number of iterations $B$.  We estimate the variance of $\theta_{1,1}$ using the bootstrap samples via the sample variance of $\{\theta_{1,1}^{*(b)}\}_{b=1}^B$
\begin{equation*}
\hat{V} = \frac{1}{B-1} \sum_{b=1}^B \left(\theta_{1,1}^{*(b)} - \frac{1}{B}\sum_{b'=1}^B \theta_{1,1}^{*(b')} \right)^2.
\end{equation*}
We form a $95\%$ confidence interval using a normal approximation
\begin{equation*}
(\hat{\theta}_{1,1} - 1.96 \cdot \sqrt{\hat{V}}, \hat{\theta}_{1,1} + 1.96 \cdot \sqrt{\hat{V}}).
\end{equation*}
Other approaches such as the percentile method described in \cite{Chen2022StatisticalOptima} will also be equivalent asymptotically since the bootstrap distribution will be asymptotically normal.  The key advantage of this method is that it avoids the difficulty of obtaining a closed form expression for the Fisher information matrix, which can be difficult for complex mixture models.  The process of the bootstrap is similar to that of attaining a point estimate because we run Algorithm \ref{alg:marEM} repeatedly many times.  However, unlike obtaining a point estimate, we use the same initialization point on many different data sets (bootstrap samples) rather than many different initialization points on the same data set.  \cite{WhiteMurphy2014} have also discussed how bootstrap methods can struggle when parameters are close to the boundary, so we recommend that post-hoc visualizations are performed to ensure that label-switching has not occurred.  In our data analysis and simulations, the mixture components are well-separated, and in practice, we order the groups based on the parameter estimates to avoid the label switching problem.

\section{Cluster Analysis}

\label{sect:clustering}

A feature of the mixture model that we propose is that we are able to
cluster individuals into different groups according to the predictive probability.
This allows us to discover latent groups inside our data.
We assume that the mixture of binomial product experts model has already been fit successfully using Algorithm \ref{alg:marEM}, and we have an estimator $(\hat{\bbeta}, \hat{\btheta})$.

Each mixture component can be thought of as a cluster, so the method we propose can be used to perform model-based clustering \citep{raftery2019}.  To simplify the problem, we first consider the case where we do not have any missing outcome variables.  That is, the data we have is only latent incomplete.  For each observation $i$, define the probabilities as follows
\begin{equation}
\hat{\pi}_k(\bX_i, \bY_i) := P(Z=k|\bX_i,\bY_i;\hat{\bbeta}, \hat{\btheta}) = \frac{w_k(\bX_i;\hat{\beta}_k) p_k(\bY_i; \hat{\theta}_k)}{\sum_{k'=1}^K w_{k'}(\bX_i;\hat{\beta}_{k'}) p_{k'}(\bY_i; \hat{\theta}_{k'})}.
\end{equation}
The quantity $\hat{\pi}_k(\bX_i, \bY_i)$ is the predictive probability that individual $i$ is from component $k$. 
We compute the probability of belonging to each class given the data $\bX$ and $\bY$, leading to
a probability vector $\hat{\pi}(\bX_i, \bY_i) = (\hat{\pi}_1(\bX_i, \bY_i), \ldots, \hat{\pi}_K(\bX_i, \bY_i))\in \mathbb{R}^K$.
It is straightforward to create a cluster assignment of each individual by assigning each observation to the component with the maximum probability. Namely, we assign individual $i$ to the cluster with
$$
\hat C_i = \argmax_k \hat{\pi}_k(\bX_i, \bY_i).
$$
In the context of the neuropsychological test scores, we can interpret the latent class as a measure of latent cognitive ability over multiple cognitive domains.  An individual's cognitive ability is a complex summary of many different attributes that we hope to measure using the neuropsychological test scores.  The latent class can be predicted using the test scores and the baseline covariates.

When the data is missing, we perform clustering using the observed data.  For any $r\in\mathcal{R}$, note that there is a closed form expression for $P(Z=k|\bX,\bY_r)$.  Define the following quantity
\begin{align*}
\hat{\pi}_{k,r}(\bX_i, \bY_{i,r}) := P(Z=k|\bX_i,\bY_{i,r};\hat{\bbeta},\hat{\btheta}) = \frac{w_k(\bX_i;\hat{\beta}_k) p_{k,r}(\bY_{i,r}; \hat{\theta}_k)}{\sum_{k'=1}^K w_{k'}(\bX_i;\hat{\beta}_{k'}) p_{k',r}(\bY_{i,r}; \hat{\theta}_{k'})}.
\end{align*}
For observation $i$ and every $k\in[K]$, we compute $\hat{\pi}_{k,\bR_i}(\bX_i, \bY_{i,\bR_i})$.  This quantity can be interpreted as the estimated probability of arising from component $k$ given the data $\bX_i$ and $\bY_{i,\bR_i}$.  When there are missing outcome variables, we compute these probabilities explicitly and use them to perform clustering.  We summarize our overall proposed methodology in Figure \ref{fig:flowchart}.  From the observed data, we run Algorithm \ref{alg:marEM} (which comprises Algorithms \ref{alg:latentEM} and \ref{alg:MI}) to obtain a point estimate.  We use the point estimate and the observed data in the Algorithm \ref{alg:bootstrap} to construct confidence intervals using the bootstrap.  Cluster assignments can also be obtained using the point estimates.

\begin{figure}
	\centering
	\includegraphics[width=14cm]{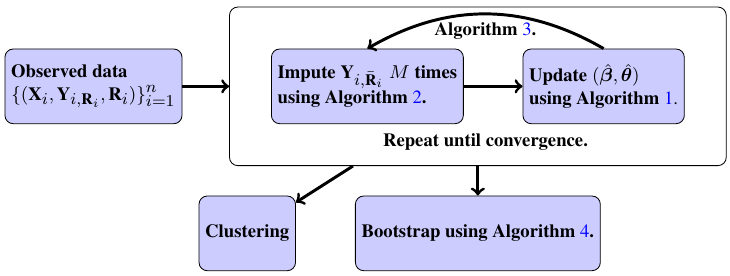}
\caption{This flowchart describes overall inference procedure.} \label{fig:flowchart}
\end{figure}


\section{Simulations}

\label{sect:simulation}

We now examine the performance of our method in finite samples using simulated data.  We have $p=2$ covariates, $d=4$ outcome variables, and $K=3$ mixtures.  Note that $\eta=2$, $\eta=2.5$, and $\eta=\infty$ correspond to approximately 75\%, 85\%, and 100\% complete cases.  We generate the data through a mixture of binomial product experts model $p(y|x)$, and the parameter $\eta$ affects the strength of the missingness via the selection model $P(R=r|x,y)$.  Further details on the data generating process and the role of $\eta$ are described in Appendix \ref{appendix:simdetails}.  Each of the simulation studies is evaluated using the mean-squared error $\text{MSE}_{\btheta} = \E[\lVert\hat{\btheta}-\btheta\rVert_2]$ and $\text{MSE}_{\bbeta} = \E[\lVert\hat{\bbeta}-\bbeta\rVert_2]$.  For the simulations in Table \ref{table:sim}, we estimate the mean-squared error by simulating $U=1000$ data sets in each scenario and average the $L_2$ norm of the error over all the data sets.  The estimated MSEs are calculated as follows

\begin{equation}
	\widehat{\text{MSE}}_{\bbeta} = \frac{1}{U}\sum_{i=1}^U \sum_{k\neq1} \lVert\hat{\beta}_k^{(i)}-\beta_k\rVert_2, \quad \widehat{\text{MSE}}_{\btheta} = \frac{1}{U}\sum_{i=1}^U \sum_{k} \lVert\hat{\theta}_k^{(i)}-\theta_k\rVert_2,
\end{equation}
where the superscript $(i)$ denotes the estimated parameter of the $i$th generated data set.  We also construct 95\% confidence intervals for each parameter and report the coverage over all the $U=1000$ data sets.  We check convergence of Algorithms \ref{alg:latentEM} and \ref{alg:marEM} by calculating the $L_2$ norm of the difference between the current and old parameter estimates and comparing against a set tolerance of $\epsilon=10^{-4}$.  In the first column of Table \ref{table:sim}, for a given $\eta$, we observe that the estimated MSE decreases at a roughly linear rate as we increase the sample size, which is expected in the parametric setting.  This demonstrates that the estimation procedures for both the no missing data (Algorithm \ref{alg:latentEM}) and missing data cases (Algorithm \ref{alg:marEM}) are consistent.  The MSE for $\bbeta$ is higher when there is missingness than when there is no missingness, which suggests that the estimating $\bbeta$ becomes more difficult.

\begin{table}
	\caption{These results are the estimated MSEs and estimated coverage after imputing $M=20$ times from our model and running Algorithms \ref{alg:marEM} and \ref{alg:bootstrap} for $U=1000$ replicates.  The first row contains results for the estimation of $\btheta$.  The second row contains results for the estimation of $\bbeta$.}
	\label{table:sim}
	\centering
	\scalebox{0.9}{
		\begin{tabular}{cccc|}
			\multicolumn{4}{c}{$\textbf{MSE}_{\btheta} (\times 100)$}                                                                                                               \\
			\cline{2-4}
			\multicolumn{1}{c|}{}        & \multicolumn{3}{c|}{Sample Size}                                                                        \\
			\cline{1-4}
			\multicolumn{1}{|c|}{}        & \multicolumn{1}{c|}{$n=500$} & \multicolumn{1}{c|}{$n=1000$} & \multicolumn{1}{c|}{$n=2000$} \\
			\hline
			\multicolumn{1}{|c|}{$\eta=\infty$} & \multicolumn{1}{c|}{0.203}        & \multicolumn{1}{c|}{0.099}         & \multicolumn{1}{c|}{0.05}                \\
			\multicolumn{1}{|c|}{$\eta=2.5$} & \multicolumn{1}{c|}{0.192}        & \multicolumn{1}{c|}{0.093}         & \multicolumn{1}{c|}{0.047}         \\
			\multicolumn{1}{|c|}{$\eta=2$} & \multicolumn{1}{c|}{0.203}        & \multicolumn{1}{c|}{0.104}         & \multicolumn{1}{c|}{0.054}              
		\end{tabular}
		
		\hspace{10pt}
		
		\begin{tabular}{cccc|}
			\multicolumn{4}{c}{\textbf{Estimated Coverage of} $\theta$s}                                                                                                               \\
			\cline{2-4}
			\multicolumn{1}{c|}{}        & \multicolumn{3}{c|}{Sample Size}                                                                        \\
			\cline{1-4}
			\multicolumn{1}{|c|}{}        & \multicolumn{1}{c|}{$n=500$} & \multicolumn{1}{c|}{$n=1000$} & \multicolumn{1}{c|}{$n=2000$} \\
			\hline
			\multicolumn{1}{|c|}{$\eta=\infty$} & \multicolumn{1}{c|}{0.946}        & \multicolumn{1}{c|}{0.948}         & \multicolumn{1}{c|}{0.950}                \\
			\multicolumn{1}{|c|}{$\eta=2.5$} & \multicolumn{1}{c|}{0.943}        & \multicolumn{1}{c|}{0.943}         & \multicolumn{1}{c|}{0.945}         \\
			\multicolumn{1}{|c|}{$\eta=2$} & \multicolumn{1}{c|}{0.939}        & \multicolumn{1}{c|}{0.939}         & \multicolumn{1}{c|}{0.932}           \\
		\end{tabular}

	}
	
	\vskip 10pt
	
	\scalebox{0.9}{
		\begin{tabular}{cccc|}
			\multicolumn{4}{c}{$\textbf{MSE}_{\bbeta} (\times 100)$}                                                                                                               \\
			\cline{2-4}
			\multicolumn{1}{c|}{}        & \multicolumn{3}{c|}{Sample Size}                                                                        \\
			\cline{1-4}
			\multicolumn{1}{|c|}{}        & \multicolumn{1}{c|}{$n=500$} & \multicolumn{1}{c|}{$n=1000$} & \multicolumn{1}{c|}{$n=2000$} \\
			\hline
			\multicolumn{1}{|c|}{$\eta=\infty$} & \multicolumn{1}{c|}{12.7}        & \multicolumn{1}{c|}{6.05}         & \multicolumn{1}{c|}{3.08}                \\
			\multicolumn{1}{|c|}{$\eta=2.5$} & \multicolumn{1}{c|}{41.5}        & \multicolumn{1}{c|}{20.6}         & \multicolumn{1}{c|}{10.5}         \\
			\multicolumn{1}{|c|}{$\eta=2$} & \multicolumn{1}{c|}{40.5}        & \multicolumn{1}{c|}{20.0}         & \multicolumn{1}{c|}{10.1}              
		\end{tabular}
		
		\hspace{10pt}
		
		\begin{tabular}{cccc|}
			\multicolumn{4}{c}{\textbf{Estimated Coverage of} $\beta$s}                                                                                                               \\
			\cline{2-4}
			\multicolumn{1}{c|}{}        & \multicolumn{3}{c|}{Sample Size}                                                                        \\
			\cline{1-4}
			\multicolumn{1}{|c|}{}        & \multicolumn{1}{c|}{$n=500$} & \multicolumn{1}{c|}{$n=1000$} & \multicolumn{1}{c|}{$n=2000$} \\
			\hline
			\multicolumn{1}{|c|}{$\eta=\infty$} & \multicolumn{1}{c|}{0.949}        & \multicolumn{1}{c|}{0.951}         & \multicolumn{1}{c|}{0.947}                \\
			\multicolumn{1}{|c|}{$\eta=2.5$} & \multicolumn{1}{c|}{0.952}        & \multicolumn{1}{c|}{0.951}         & \multicolumn{1}{c|}{0.944}         \\
			\multicolumn{1}{|c|}{$\eta=2$} & \multicolumn{1}{c|}{0.956}        & \multicolumn{1}{c|}{0.951}         & \multicolumn{1}{c|}{0.950}              
		\end{tabular}
		
	}
	
\end{table}

\begin{table}
	\caption{These results are the estimated MSEs after imputing $M=20$ times using \textit{mice} for $U=200$ replicates.  The left table contains results for the estimation of $\btheta$.  The right table contains results for the estimation of $\bbeta$.}
	\label{table:MICEsim}
	\centering
	\scalebox{0.9}{
		
		\begin{tabular}{cccc|}
			\multicolumn{4}{c}{$\textbf{MSE}_{\btheta} (\times 100)$}                                                                                                               \\
			\cline{2-4}
			\multicolumn{1}{c|}{}        & \multicolumn{3}{c|}{Sample Size}                                                                        \\
			\cline{1-4}
			\multicolumn{1}{|c|}{}        & \multicolumn{1}{c|}{$n=500$} & \multicolumn{1}{c|}{$n=1000$} & \multicolumn{1}{c|}{$n=2000$} \\
			\hline               
			\multicolumn{1}{|c|}{$\eta=2.5$} & \multicolumn{1}{c|}{0.183}        & \multicolumn{1}{c|}{0.096}         & \multicolumn{1}{c|}{0.045}         \\
			\multicolumn{1}{|c|}{$\eta=2$} & \multicolumn{1}{c|}{0.194}        & \multicolumn{1}{c|}{0.105}         &
			\multicolumn{1}{c|}{0.045} 
		\end{tabular}
		
	}
	\scalebox{0.9}{
		\begin{tabular}{cccc|}
			\multicolumn{4}{c}{$\textbf{MSE}_{\bbeta} (\times 100)$}                                                                                                               \\
			\cline{2-4}
			\multicolumn{1}{c|}{}        & \multicolumn{3}{c|}{Sample Size}                                                                        \\
			\cline{1-4}
			\multicolumn{1}{|c|}{}        & \multicolumn{1}{c|}{$n=500$} & \multicolumn{1}{c|}{$n=1000$} &
			\multicolumn{1}{c|}{$n=2000$} \\ 
			\hline
			\multicolumn{1}{|c|}{$\eta=2.5$} & \multicolumn{1}{c|}{39.4}        & \multicolumn{1}{c|}{19.0}         &
			\multicolumn{1}{c|}{10.2}  \\        
			\multicolumn{1}{|c|}{$\eta=2$} & \multicolumn{1}{c|}{46.0}        & \multicolumn{1}{c|}{19.6}         &
			\multicolumn{1}{c|}{11.6}         
		\end{tabular}
		
	}
	
\end{table}

In the second column of Table \ref{table:sim}, we generally observe nominal coverage as well.  There is slight undercoverage of $\btheta$ when $\eta=2$.  We attribute this to the fact that we do not have a closed-form expression for the MLE and there are many approximations in the algorithms.  For example, Algorithm \ref{alg:marEM} is a stochastic EM algorithm because it relies on a prespecified number of imputations $M$.  Increasing the number of imputations will decrease the Monte Carlo error at the cost of computational time.  It is possible that the gradient of the log-likelihood surface can be fairly small in the neighborhood of the true MLE, leading to early stopping of the algorithms and returning a near-solution but not the true MLE.  Moreover, when there is missing data, the effective sample size is lower than the number of observations $n$, which can be another reason for the slight undercoverage in finite samples.  For most cases, the coverage is close to nominal, which suggests that the bootstrap procedure and the choice of initializing it at the original estimate returned by Algorithm \ref{alg:marEM} works in practice.

\begin{table}[!b]
	\caption{For the simulations in Tables \ref{table:sim} and \ref{table:MICEsim}, we report the average \textbf{computation time in minutes} and the standard deviation in parentheses across all the randomly generated data sets.  These were performed on a cluster using 48 cores.  The first table contains the times for obtaining both a point estimate and a confidence interval using Algorithms \ref{alg:marEM} and \ref{alg:bootstrap}.  The second table contains the times for only obtaining a point estimate with \textit{mice} and Algorithm \ref{alg:latentEM}.\\}
	\label{table:comptime-sim}
	\centering
\scalebox{0.80}{
	\begin{tabular}{cccc|}
		\multicolumn{4}{c}{\parbox{10cm}{\centering Algorithms \ref{alg:marEM} and \ref{alg:bootstrap} (Point Estimate and Bootstrap CI)}}
		\\
		\cline{2-4}
		\multicolumn{1}{c|}{} & \multicolumn{3}{c|}{Sample Size} \\
		\cline{1-4}
		\multicolumn{1}{|c|}{} & \multicolumn{1}{c|}{$n=500$} & \multicolumn{1}{c|}{$n=1000$} & \multicolumn{1}{c|}{$n=2000$} \\
		\hline
		\multicolumn{1}{|c|}{$\eta = 2.5$} & \multicolumn{1}{c|}{2.0 (2.1)}  & \multicolumn{1}{c|}{2.8 (0.5)}  & \multicolumn{1}{c|}{5.5 (1.1)} \\
		\multicolumn{1}{|c|}{$\eta = 2$}   & \multicolumn{1}{c|}{2.4 (1.9)}  & \multicolumn{1}{c|}{3.9 (3.5)}  & \multicolumn{1}{c|}{6.2 (1.9)} \\
	\end{tabular}

	\hspace{0.15in}

	\begin{tabular}{cccc|}
		\multicolumn{4}{c}{\parbox{8cm}{\centering \textit{mice} and Algorithm \ref{alg:latentEM} (Point Estimate)}}
		\\
		\cline{2-4}
		\multicolumn{1}{c|}{} & \multicolumn{3}{c|}{Sample Size} \\
		\cline{1-4}
		\multicolumn{1}{|c|}{} & \multicolumn{1}{c|}{$n=500$} & \multicolumn{1}{c|}{$n=1000$} & \multicolumn{1}{c|}{$n=2000$} \\
		\hline
		\multicolumn{1}{|c|}{$\eta = 2.5$} & \multicolumn{1}{c|}{8.1 (0.91)}  & \multicolumn{1}{c|}{14.2 (0.83)}  & \multicolumn{1}{c|}{28.2 (3.8)} \\
		\multicolumn{1}{|c|}{$\eta = 2$}   & \multicolumn{1}{c|}{8.1 (0.76)}  & \multicolumn{1}{c|}{14.5 (0.89)}  & \multicolumn{1}{c|}{27.4 (2.1)} \\
	\end{tabular}
}
\end{table}

For an additional comparison, we conduct a second simulation using the same data generating process.   We first perform imputation using the \textit{mice} R package \citep{mice}, generating $M = 20$ imputed datasets. We then fit our model using Algorithm \ref{alg:latentEM} and compute the MSEs.  The results are reported in Table \ref{table:MICEsim}.
	
	All simulations in this section were run on a cluster with 48 cores and a parallelized implementation. Our first observation is that this two-step approach—imputing with \textit{mice} followed by Algorithm \ref{alg:latentEM}—is substantially slower than Algorithms \ref{alg:marEM} and \ref{alg:bootstrap}. As a result, we limit the \textit{mice} experiment to $200$ total repetitions for estimating the MSE and do not report coverage results since bootstrap-based interval estimation would be prohibitively time-consuming. For reference, obtaining \textit{both} a point estimate and a bootstrapped confidence interval with our method requires approximately 3 to 6 times less computation time than obtaining \textit{only} a point estimate after imputing with the \textit{mice} method.  This can be seen in Table \ref{table:comptime-sim}.
	
	Interestingly, we observe that the MSEs from the \textit{mice} imputation and Algorithm \ref{alg:latentEM} pipeline are comparable to those from our method. However, we believe this is largely due to the well-separated mixture components in the data generating process. In general, we do not recommend using \textit{mice} in this context, as its imputation model may not be compatible with the mixture model being fit, and may not lead to asymptotically valid estimates in more challenging settings. Model uncongeniality is a known problem in the imputation literature \citep{meng1994multiple}.  In contrast, our method imputes from the same model that is estimated and fit, ensuring consistency.\\

Further simulations on the complete case data and the influence of the number of imputations are also included in Appendix \ref{appendix:add-sim}. When fitting the model on complete case data, we observe that the MSE does not decrease at a linear rate due to bias, and that the confidence intervals are not asymptotically valid. Moreover, the results appear fairly stable as we vary the number of imputations $M$ across 10, 20, and 50, suggesting that 10 imputations are already sufficient. To ensure better accuracy of our estimates while maintaining computational efficiency, we use $M=20$ imputations in the real data analysis.

\section{Application to the NACC data}

\label{sect:realdata}

As mentioned earlier, the data set that motivates our model is from the National Alzheimer's Coordinating Center (NACC) Uniform Data Set, collected from the years 2005 to 2019.  This is a longitudinal data set that comprises individuals of varying cognitive status: cognitively normal to mild cognitive impairment (MCI) to dementia.  Neuropsychological test scores are typically assessed annually with some baseline covariates collected upon entry to the study.  However, there is also some missingness present in the outcome variables.  The missingness can be due to various reasons.  For example, some tests may no longer be administered after a given time because they have been replaced with another one.  On the other hand, some individuals may have test scores missing because of a data recording error or they are too sick to take more tests.

\subsection{Description of outcome variables and covariates}

Our primary goal is to measure the cognitive ability of the Alzheimer's disease patients while measuring its association with baseline covariates.  Following similar analysis in \cite{brenowitz2017mixed}, we focus on four cognitive domains: long-term memory (episodic memory), attention, language, and executive function. Specifically, we used the Logical Memory Story A immediate and delayed recall to assess memory, Digit Span Forward and Backward tests to evaluate attention, Animal listing test to measure language ability, and Trail Making Tests Parts A and B to measure executive function.  Additionally, we used the Mini-Mental State Examination (MMSE) to assess overall cognitive impairment, resulting in a total of eight outcome variables.  The Trail A and B tests are assessed based on the time it takes to complete a given task, so a higher score is considered worse.  On the other hand, for the other tests, higher scores are indicative of stronger cognitive ability.  

We include four baseline covariates in this study: age, education, sex, and race. Age was kept on a yearly scale, and education was dichotomized based on whether an individual had obtained a college degree. Of particular interest was the association between education and cognitive decline. This analysis was motivated by the cognitive reserve hypothesis, which suggests that some mechanisms can provide individuals more resilience to cognitive decline.  We would like to explore this hypothesis to see if individuals with higher levels of education are more resistant to cognitive impairment \citep{Meng2012EducationAD, ThowMeganE2018Feic}.

For each individual, we have a CDR (Clinical Dementia Rating) score, which is assigned by a medical professional.  The CDR score is based on a clinician's judgment on an individual's cognitive ability
and is the standard procedure to determine whether someone has dementia or not.
A CDR score 0 refers to cognitively normal, a score 0.5 refers to a mild cognitive impairment (MCI) state,
and a score of at least 1 refers to dementia. 
Note that the CDR score is all based on a clinician's judgment, which is very different from
the neuropsychological tests (exam-based).  This score takes values in the set $\{0,0.5,1,2,3\}$.  Here, 0 indicates normal cognitive ability, 0.5 is mild cognitive impairment (MCI), 1 is mild dementia, 2 is moderate dementia, and 3 is severe dementia.  Thus, we expect that a healthy group would contain individuals with lower CDR scores than an unhealthier group.  We will use the CDR score to interpret some of our results, but we emphasize that we \textit{do not} use this variable in the model.

Since there are 8 outcome variables, there are up to $2^8 = 256$ missing data patterns.  In this data set, we observe 93 missing data patterns with 7 patterns comprising over 95\% of the data.  In particular, there are $n_{cc}=25041$ fully observed cases and $n=41181$ total observations, so that approximately only 60\% of the individuals have completely observed outcomes.

\begin{figure}
	\centering
\includegraphics[width=14cm]{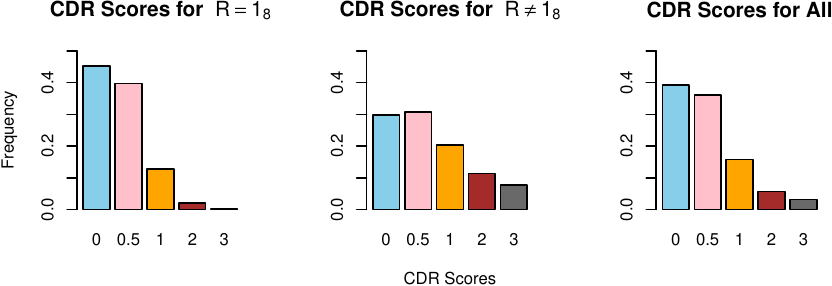}
\caption{The CDR score distributions for the complete cases, the individuals missing at least one outcome variable, and the entire data set are provided in the left, middle, and right panels, respectively.}
\label{fig:CDRdist}
\end{figure}

In Figure \ref{fig:CDRdist}, we plot the CDR distribution for the complete cases, the individuals missing at least one outcome variable, and the entire data set, respectively.  If the data was missing completely at random, we would expect the CDR score distribution to be mostly constant across all missing patterns of the data.  However, we observe that for the complete cases, the dementia group (CDR score of 1, 2, or 3) makes up less than 20\%, but the dementia group is close to 40\% of the data that has at least one missing outcome variable.  This suggests that the dementia group is severely underrepresented when we only perform analysis on the complete cases.  Thus, from the plots, we can visually postulate the MCAR assumption to be unreasonable.

We select the model using prior knowledge.  Since the scientific community previously determined that five levels for the CDR score was appropriate, we also choose the same number.  A similar number of groups will help with interpretability, and five is small enough that the number of parameters is manageable.

\subsection{Point estimates and confidence intervals}

The data analysis was conducted on a 2022 Macbook Air with 8 cores, and it took about a day to fit the model and complete the bootstrap with 1000 bootstrap samples.  In Table \ref{table:testscoremeans}, we report the point estimates of the test score means and the corresponding standard errors in parentheses.  There are a few takeaways.  As reported, Classes 1 and 5 can be interpreted as the most and least healthy individuals, respectively.  Note that in 6 out of 8 tests (omitting TRAIL A and TRAIL B), a higher score corresponds to better performance.  On the other hand, TRAIL A and TRAIL B are scored using time-to-completion in seconds, so a higher score corresponds to a worse performance.  For every test aside from TRAIL B, we observe a clear monotonic behavior in the test score means from Class 1 to Class 5, which suggests that the model is reasonable and fairly interpretable.  For TRAIL B, Classes 4 and 5 have very similar mean test scores that are close to the maximum possible test score of 300.  Since a test score of 300 indicates that the individual timed out and did not complete the test, this suggests that individuals from Classes 4 and 5 have comparable performance on the TRAIL B test.

Additionally, the results also suggests the different tests may have value at distinguishing individuals from different latent cognitive classes.  For example, TRAIL A and TRAIL B are routinely used to measure executive function.  Classes 4 and 5 represent the most unhealthy individuals in the population, but individuals in both perform similarly on TRAIL B.  On the other hand, the mean test scores in Classes 4 and 5 for TRAIL A are very different, which suggests that TRAIL A can be a good discriminator between very unhealthy and moderately unhealthy individuals.  As TRAIL B is a more complex task than TRAIL A, this seems to agree with our intuition.

\begin{table}
\caption{These tables contain point estimates of the test score means for each latent class.  In the title of each column, we report the maximum possible score for a given test in parentheses. In each cell, the standard errors are reported in parentheses.}
\label{table:testscoremeans}
\centering
\scalebox{0.85}{
	\begin{tabular}{cccccc}
		&         & \multicolumn{4}{c}{Outcome Variables}                                          \\ \\
		&         & \textbf{MMSE (30)} & \textbf{LOGIMEM (25)} & \textbf{MEMUNITS (25)} & \textbf{DIGIF (12)}\\
		\multirow{5}{*}{\rotatebox[origin=c]{90}{Class}} & 1 & 28.9 (0.01)  & 13.4 (0.04) & 12.0 (0.05)  & 8.8  (0.02) \\
		& 2 & 26.8 (0.03) & 9.1 (0.05)  & 7.1 (0.06)  & 7.7  (0.02)  \\
		& 3 & 25.3  (0.06) & 7.6  (0.08)  & 5.5  (0.08)  & 7.1 (0.03) \\
		& 4 & 22.5 (0.06) & 5.3  (0.05)  & 3.5  (0.05)  & 6.6 (0.03)  \\
		& 5 & 14.4 (0.12)  & 3.1 (0.05)   & 1.8  (0.04)  & 5.3  (0.04) 
	\end{tabular}
}
\vskip 10pt
\scalebox{0.85}{
	\begin{tabular}{cccccccccc}
		&         & \multicolumn{4}{c}{}                                          \\
		&         & \textbf{DIGIB (12)} & \textbf{ANIMALS (77)} & \textbf{TRAILA (150)} & \textbf{TRAILB (300)} \\
		\multirow{5}{*}{\rotatebox[origin=c]{90}{Class}} & 1 & 7.2  (0.05)  & 20.9 (0.05)  & 28.4 (0.09)  & 65.9 (0.2)  \\
		& 2 & 5.7  (0.02)  & 15.9  (0.05) & 44.2 (0.17)  & 117.9 (0.43) \\
		& 3 & 4.9  (0.03)  & 13.5  (0.07)  & 57.2 (0.39)  & 209.7  (0.84)\\
		& 4 & 4.2  (0.02)  & 11.0 (0.06)  & 66.7 (0.36)  & 298.5 (0.08) \\
		& 5 & 2.7  (0.03)  & 6.9 (0.07)   & 145.3 (0.20) & 295.2 (0.09)
	\end{tabular}
}
\end{table}

In Table \ref{table:covariates}, we report the point estimates of the coefficients in the logistic regressions.  One initial observation is that the coefficients of Age and Education are all positive and negative at the $\alpha=0.05$ significance level, respectively, across all classes.  This implies that increasing age is associated with lower cognitive ability and higher education may provide some protection against dementia.  The latter is known as the cognitive reserve hypothesis in the context of education \citep{Meng2012EducationAD, ThowMeganE2018Feic}.  We can compare the magnitudes of the coefficient of Education to the coefficient of Age.  For example, in Class 2, we observe $0.62/0.06 \approx 10$, and in Class 5, we observe $1.14/0.05\approx 23$.  This suggests that having a college degree or higher education may be equivalent to being 10-20 years younger in terms of cognitive ability.

\begin{table}
\caption{This table contains point estimates for the coefficients for each of the covariates and each latent class.  In each cell, the standard errors are reported in parentheses.}
\label{table:covariates}
\centering
\scalebox{0.85}{
	\begin{tabular}{cccccccccc}
		&         & \multicolumn{4}{c}{Covariates}                                          \\
		&         & \textbf{Age} & \textbf{Education} & \textbf{Race} & \textbf{Sex} \\
		\multirow{4}{*}{\rotatebox[origin=c]{90}{Class}} & 2 & 0.06 (0.001)  & -0.62 (0.028) & -0.37 (0.026)  & 0.72  (0.037)  \\
		& 3 & 0.07 (0.002) & -0.9 (0.039)  & 0.45 (0.034)  & 0.93  (0.043) \\
		& 4 & 0.06 (0.002) & -1.06  (0.035)  & 0.36  (0.032)  & 0.79 (0.04) \\
		& 5 & 0.05 (0.002) & -1.14  (0.0041)  & 0.41 (0.0035)  & 0.93 (0.0041)
	\end{tabular}
}
\end{table}

\subsection{Latent classes and clustering}

We report two clustering results using models fit on the complete cases (Figure \ref{fig:pie_CC}) and the entire data (Figure \ref{fig:pie_MAR}).  For each set of clustering results, we visualize the clusters by plotting the CDR score composition of each of the five clusters via barplots.  Furthermore, we report the mean of the CDR scores for each group, which gives a rough estimate on the group-specific clinical cognitive ability.  From the complete case results (CC) in Figure \ref{fig:pie_CC}, we generally see that the first clusters are primarily composed of healthy individuals as individuals with CDR scores of 0 take up larger proportions.  As we progress to the middle latent classes, healthier individuals take up fewer proportions and individuals with mild cognitive impairment (CDR score of 0.5) take up larger proportions.  Then, in Group 5, the majority of the individuals have some form of dementia.

For the clustering that used the entire data with the missing at random assumption (Figure \ref{fig:pie_MAR}), we observe a similar trend with the healthier individuals, the MCI individuals, and the dementia individuals being the dominant proportion in the earlier, middle, and later latent classes.  However, the results from clustering the entire data set are more stark.  In contrast with the complete clustering, Groups 4 and 5 contain a significantly higher number of individuals.  Moreover, our model is able to detect a group of individuals with high dementia, as evidenced by Group 5's mean CDR score of 1.55.  Unlike the complete case clustering, Group 5 has almost 75\% of individuals with some form of dementia.  Many of these individuals would be omitted from the analysis if the missing data was not accounted for.
Thus, accounting for the missingness leads to a clustering result with a stronger correlation with clinical assessments (CDR score). 

\begin{figure}[!h]
	\centering
\scalebox{0.88}{\includegraphics[height=3.5cm]{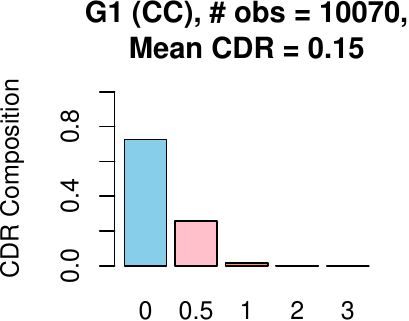}}
\hspace{0.1in}
\scalebox{0.88}{\includegraphics[height=3.5cm]{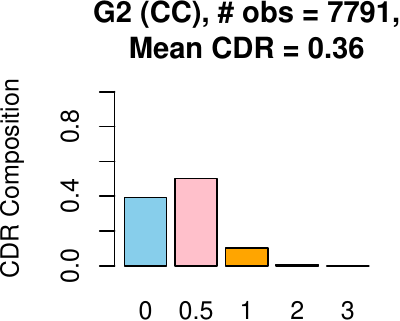}}
\hspace{0.1in}
\scalebox{0.88}{\includegraphics[height=3.5cm]{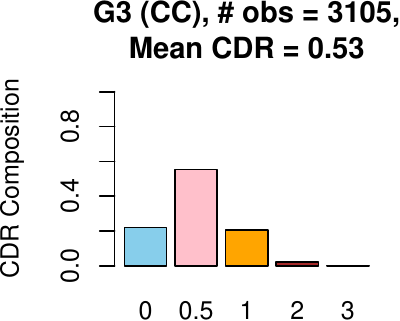}}
\vskip 20pt
\scalebox{0.88}{\includegraphics[height=3.5cm]{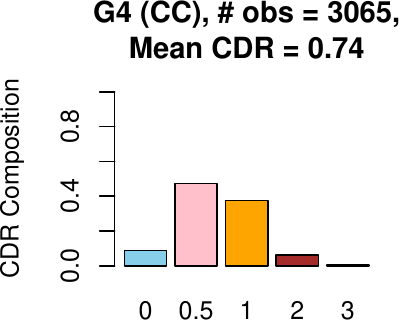}}
\hspace{0.1in}
\scalebox{.9}{\includegraphics[height=3.5cm]{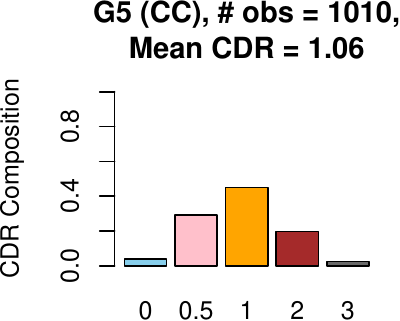}}
\centering
\caption{Clustering on complete data only. These barplots summarize the composition of each of the five latent groups.  We order the groups from most healthy to least healthy.  This is reflected in the mean CDR score of each group.}
\label{fig:pie_CC}
\end{figure}

\begin{figure}[!h]
	\centering
\scalebox{0.88}{\includegraphics[height=3.5cm]{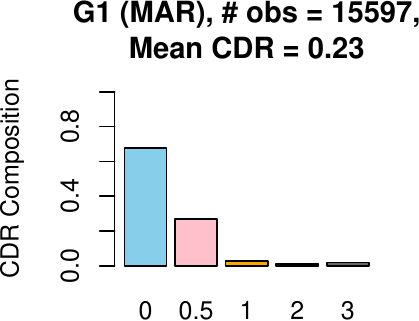}}
\hspace{0.1in} 
\scalebox{0.88}{\includegraphics[height=3.5cm]{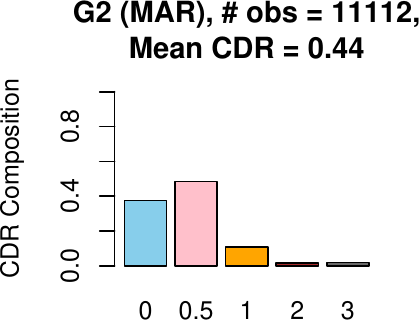}}
\hspace{0.1in}
\scalebox{0.88}{\includegraphics[height=3.5cm]{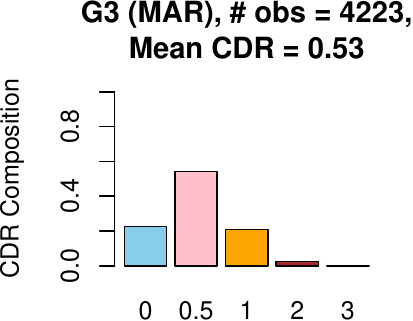}}
\vskip 20pt
\scalebox{0.88}{\includegraphics[height=3.5cm]{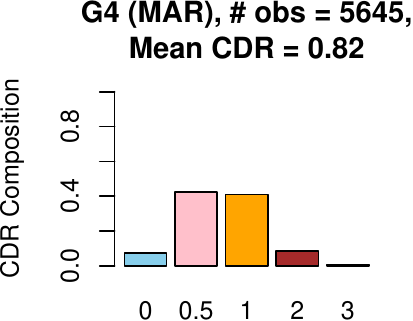}}
\hspace{0.1in}
\scalebox{0.88}{\includegraphics[height=3.5cm]{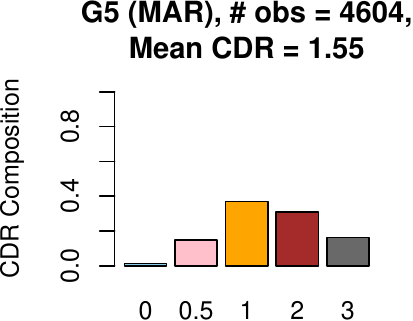}}
\centering
\caption{Clustering on the entire data with MAR assumption. These barplots summarize the composition of each of the five latent groups.  We order the groups from most healthy to least healthy.  This is reflected in the mean CDR score of each group.}
\label{fig:pie_MAR}
\end{figure}

\section{Discussion}

\label{sect:discussion}

In this paper, we have proposed a mixture of binomial experts model for modeling neuropsychological test score data.  This model builds on classical ideas from the latent variable and mixture of experts literature.  Through this model, we are able to construct a latent representation of an individual's cognitive ability using their test scores and relate it to baseline covariates.  Because the Alzheimer's disease data set is both enriched with multivariate information and plagued with missing data, we address both of these issues.  We outline how to perform estimation and inference under the MAR assumption. 
Further comments on simulations are included in Appendix \ref{sect:simulations}.

There are several avenues for extending this paper.  Our presented model has two key components: the weights $w_k(x;\beta_k)$ and the component probabilities $p_k(y;\theta_k)$.  Due to the local independence assumption, the component probabilities $p_k(y;\theta_k)$ reduce to a product of marginal distributions, which is a binomial product in this paper.  One straightforward way to generalize this work is modifying these two working models to other parametric families.  In our setting, the outcome data is discrete and bounded because we have neuropsychological test scores.  However, often, outcome data can come in a variety of different forms such continuous, discrete, and mixed, as well as bounded and unbounded.  For instance, in continuous and count data, respectively, the Gaussian and Poisson distributions may be applicable.  Access to individual question level data may allow us to leverage ideas from item response theory and the Rasch model \citep{rasch1960probabilistic}.  Additionally, since there is longitudinal information in the data, there may be a way to incorporate time in the model, so that model-based clustering can be performed on entire trajectories rather than individuals at given a time point.  There has been previous work on mixture of experts models applied to time series and longitudinal data \citep{waterhouse1995, huerta2003, tang2016}, and it would be interesting to extend these frameworks to a model for multivariate longitudinal discrete data.  Model selection is also another avenue for future exploration as we selected our model using prior knowledge.  We included comments on AIC and BIC in Appendix \ref{appendix:modelselect}. Also, since we have 33 ADRC centers collecting data, we may introduce a random effect model in the weight to handle the effect from different centers.

The MAR assumption is fundamental to our analysis of the NACC data as well as the inference procedure (Algorithm \ref{alg:marEM}), but it is not the only option.  Missing not at random (MNAR) assumptions encompass a rich class of potentially plausible assumptions.  It would be interesting to explore MNAR assumptions in the context of latent variable modeling using methods such as pattern graphs \citep{chen2022, cheng2022handling, suen2021multistage}.  Additionally, perturbing the MAR as a form a sensitivity analysis remains largely an open question.  We note that departures from missing at random would likely result in a significant modification of the elegant estimation procedure outlined in the paper because the ignorability condition would no longer hold.  Missing not random assumptions often require stronger modeling assumptions because the relationship between $\bY$ and $\bR$ needs to be modeled directly.  We leave this for future work.


\clearpage

\newpage

\appendix
\section{Detailed Derivation of the EM Algorithm Update Equations}

\label{appendix:EMderivation}

There are two EM algorithms described in this paper.  The inner EM algorithm (Algorithm \ref{alg:latentEM}) is used to estimate the mixture of binomial product experts model while the outer EM algorithm (Algorithm \ref{alg:marEM}) is used to fit the aforementioned model under a missing at random assumption.  When there is no missingness, Algorithm \ref{alg:marEM} reduces to Algorithm \ref{alg:latentEM}.   We discuss the details of both algorithms and the multiple imputation algorithm in the following three subsections.

\subsection{Latent Variable EM}

The E-step is easily obtained by calculating the $Q_{\text{LC}, n}$ function.  Since the only variable that is not conditioned on is $Z$, this is simply equivalent to changing all of the indicator variables in $\ell_{\text{LC},n}$ to conditional probabilities on $\bX$ and $\bY$.

Thus, the $Q_{\text{LC}, n}$ function writes as 
\begin{align}
	&Q_{\text{LC}, n}(\bbeta, \btheta | \bX_{1:n}, \bY_{1:n}; \hat{\bbeta}^{(t)}, \hat{\btheta}^{(t)}) \nonumber \\
	&:= \frac{1}{n}\sum_{i=1}^n\sum_{k=1}^K \hat{\pi}_k^{(t)}(\bX_i, \bY_i)
	\log w_k(\bX_i;\beta_k) + \frac{1}{n}\sum_{i=1}^n\sum_{k=1}^K \hat{\pi}_k^{(t)}(\bX_i, \bY_i) \log p_k(\bY_i;\theta_k), \label{eq:Q-LCn}
\end{align}
where
\begin{equation*}
	\hat{\pi}_k^{(t)}(\bX_i, \bY_i) := P(Z=k|\bX_i,\bY_i;\hat{\bbeta}^{(t)}, \hat{\btheta}^{(t)}) = \frac{w_k(\bX_i;\hat{\beta}_k^{(t)}) p_k(\bY_i; \hat{\theta}_k^{(t)})}{\sum_{k'=1}^K w_{k'}(\bX_i;\hat{\beta}_{k'}^{(t)}) p_{k'}(\bY_i; \hat{\theta}_{k'}^{(t)})}.
\end{equation*}

We now derive the update equations in the M-step.  Note that we update $(\bbeta, \btheta)$ via
$$(\hat{\bbeta}^{(t+1)}, \hat{\btheta}^{(t+1)}) = \argmax_{\bbeta, \btheta} Q_{\text{LC}, n}(\bbeta, \btheta | \bX_{1:n}, \bY_{1:n}; \hat{\bbeta}^{(t)}, \hat{\btheta}^{(t)}).$$

Taking a partial derivative of \eqref{eq:Q-LCn} with respect to $\theta_{k,j}$, we obtain
\begin{align*}
	&\frac{\partial}{\partial \theta_{k,j} }\frac{1}{n}\sum_{i=1}^n \hat{\pi}_k^{(t)}(\bX_i, \bY_i) \log \left( \binom{N_j}{Y_{i,j}} \theta_{k,j}^{Y_{i,j}} (1-\theta_{k,j})^{N_j-Y_{i,j}} \right) \\
	&= \frac{\partial}{\partial \theta_{k,j} } \frac{1}{n}\sum_{i=1}^n \hat{\pi}_k^{(t)}(\bX_i, \bY_i) \cdot (Y_{i,j} \log \theta_{k,j} + (N_j-Y_{i,j}) \log(1-\theta_{k,j})) \\
	&= \frac{1}{n}\sum_{i=1}^n \hat{\pi}_k^{(t)}(\bX_i, \bY_i) \cdot \left(\frac{Y_{i,j}}{\theta_{k,j}} - \frac{N_j-Y_{i,j}}{1-\theta_{k,j}} \right).
\end{align*}
Setting this derivative equal to $0$ yields the desired update equation for $\btheta$
\begin{equation*}
	\hat{\theta}_{k,j}^{(t+1)} = \frac{\sum_{i=1}^n \hat{\pi}_k^{(t)}(\bX_i, \bY_i) \cdot \frac{Y_{i,j}}{N_j}}{\sum_{i=1}^n \hat{\pi}_k^{(t)}(\bX_i, \bY_i)}.
\end{equation*}

We now discuss the update procedure for $\bbeta$.  We see that
$$\frac{1}{n}\sum_{i=1}^n\sum_{k=1}^K \hat{\pi}_k^{(t)}(\bX_i, \bY_i) \log p_k(\bY_i;\theta_k)$$
is the log-likelihood function with weights $\hat{\pi}_k^{(t)}(\bX_i, \bY_i)$ rather than $1(Z_i=k)$, so we can apply a standard package that fits a logistic regression model using gradient ascent.  However, instead of treating the dependent variable as $(1(Z_i=1), 1(Z_i=2), \ldots, 1(Z_i=K))$, we use the dependent variable $\mathbf{W}_i(\bX_i, \bY_i) \equiv (\hat{\pi}_1^{(t)}(\bX_i, \bY_i), \hat{\pi}_2^{(t)}(\bX_i, \bY_i), \ldots, \hat{\pi}_K^{(t)}(\bX_i, \bY_i))$.

This completes the derivation of Algorithm \ref{alg:latentEM}.

\subsection{Multiple Imputation under Missing at Random}

\label{appendix:multiple-imp}

The missing at random property equivalently can be written as
\begin{align*}
	P(\bR = r | x,y) = P(\bR = r | x, y_r) &\iff p(x,y,r) = p(x,y_r,r) \cdot p(y_{\bar{r}} | y_r, x) \\
	& \iff p(y_{\bar{r}} | y_r, x, r) = p(y_{\bar{r}} | y_r, x)
\end{align*}
for all $r\in\mathcal{R}$. 

Therefore, the imputation distribution $p(y_{\bar{r}} | y_r, x, r) \equiv p(y_{\bar{r}} | y_r, x)$ is constructed using the global model $p(y | x)$ that is fit on all of the data regardless of the missing data pattern.

Algorithm \ref{alg:MI} requires the estimation of the conditional distribution $p(y_{\bar{r}} | y_r, x; \bbeta, \btheta)$ for every $r\in\mathcal{R}$.  In the following derivation of the imputation distribution, we drop the parameters and reference to the iteration $t$ for readability.  For any $r\in\mathcal{R}$, we have
\begin{align*}
	p(y_{\bar{r}} | y_r, x) &= \frac{p(y | x)}{p(y_r | x)} \\
	&= \frac{p(y | x)}{\sum_{y_{\bar{r}}} p(y | x)} \\
	&= \frac{\sum_{k=1}^K w_k(x) \prod_{j=1}^d p_{k,j}(y_j)}{\sum_{y_{\bar{r}}} \sum_{k=1}^K w_k(x) \prod_{j=1}^d p_{k,j}(y_j)} \\
	&= \frac{\sum_{k=1}^K w_k(x) \prod_{j=1}^d p_{k,j}(y_j)}{\sum_{k=1}^K w_k(x) \sum_{y_{\bar{r}}} \prod_{j=1}^d p_{k,j}(y_j)} \\
	&= \frac{\sum_{k=1}^K w_k(x) \prod_{j=1}^d p_{k,j}(y_j)}{\sum_{k=1}^K w_k(x) \prod_{r} p_{k,j}(y_j)} \\
	&= \sum_{k=1}^K \underbrace{\frac{w_k(x) \prod_{j\in r} p_{k,j}(y_j)}{\sum_{k=1}^K w_k(x) \prod_{j\in r} p_{k,j}(y_j)}}_{W_{k,r}(x,y_r)} \prod_{j\in\bar{r}} p_{k,j}(y_j) \\
	&= \sum_{k=1}^K W_{k,r}(x,y_r) \prod_{j\in\bar{r}} p_{k,j}(y_j).
\end{align*}
Therefore, with a little algebra and the local independence assumption, we can show that the conditional distribution (and imputation distribution) remains a (reweighted) mixture of binomial products distribution, whose computation is fairly straightforward.

\subsection{Missing at Random EM}

Standard maximum likelihood inference can be viewed as constructing an $M$-estimator that estimates the population MLE
$$(\bbeta_0, \btheta_0) = \argmax_{\bbeta, \btheta} \E[\log p(\bY|\bX; \bbeta, \btheta)].$$

By the law of iterated expectation, observe that
\begin{align}
	\E[\log p(\bY|\bX; \bbeta, \btheta)] &= \sum_r \E[\E \log p(\bY|\bX; \bbeta, \btheta) | \bX, \bY_r, \bR=r ] \cdot 1(\bR=r)] \nonumber \\
	&\stackrel{\text{MAR}}{=} \sum_r \E[\E \log p(\bY|\bX; \bbeta, \btheta) | \bX, \bY_r ] \cdot 1(\bR=r)]. \label{eq:pop-MAR-MLE} \\
	&= \E\left[\sum_{r\in\mathcal{R}}\E [\log p(\bY|\bX; \bbeta, \btheta) | \bX, \bY_r ] \cdot 1(\bR=r) \right].
\end{align}
By converting equation \eqref{eq:pop-MAR-MLE} to the sample form, we obtain
\begin{align}
	Q_{\text{FC},n}(\bbeta, \btheta| \bX_{1:n}, \bY_{1:n} ; \hat{\bbeta}^{(t)}, \hat{\btheta}^{(t)}) &= \frac{1}{n} \sum_{i=1}^n \sum_{r\in\mathcal{R}} Q_{\text{FC}, r}(\bbeta, \btheta | \bX_i, \bY_{i,r}; \hat{\bbeta}^{(t)}, \hat{\btheta}^{(t)}) \cdot 1(\bR_i = r), \nonumber \\
	&= \frac{1}{n} \sum_{i=1}^n Q_{\text{FC}, \bR_i}(\bbeta, \btheta | \bX_i, \bY_{i,\bR_i} ; \hat{\bbeta}^{(t)}, \hat{\btheta}^{(t)}), \label{eq:Q_FCn}
\end{align}
where $Q_{\text{FC}, r}(\bbeta, \btheta | \bX, \bY_r ; \hat{\bbeta}^{(t)}, \hat{\btheta}^{(t)}) = \E[\log p(\bY | \bX) | \bX, \bY_r ; \hat{\bbeta}^{(t)}, \hat{\btheta}^{(t)}]$.  By construction, we expect that $Q_{\text{FC},n}(\bbeta, \btheta | \bX_{1:n}, \bY_{1:n} ; \hat{\bbeta}^{(t)}, \hat{\btheta}^{(t)})$ has expectation equal to $\E[\log p(\bY | \bX; \bbeta, \btheta)]$.  Unfortunately, we are not able to compute \eqref{eq:Q_FCn} because we do not have an explicit closed form expression for $Q_{\text{FC}, r}(\bbeta, \btheta | \bX, \bY_r ; \hat{\bbeta}^{(t)}, \hat{\btheta}^{(t)})$.  Thankfully, we can approximate it stochastically.  For large enough $M$, we expect
\begin{align*}
	\frac{1}{M}\sum_{m=1}^M \log p(\tilde{\bY}^{(m;t)} | \bX; \hat{\bbeta}^{(t)}, \hat{\btheta}^{(t)}) \approx Q_{\text{FC}, r}(\bbeta, \btheta | \bX, \bY_r ; \hat{\bbeta}^{(t)}, \hat{\btheta}^{(t)}),
\end{align*}
where $\tilde{\bY}^{(m;t)} = (\bY_r, \bY_{\bar{r}}^{(m;t)})$, and for each $m$, $\bY_{\bar{r}}^{(m;t)} \stackrel{\text{iid}}{\sim} p(y_{\bar{r}} | \bY_r, \bX; \hat{\bbeta}^{(t)}, \hat{\btheta}^{(t)})$.  Thus, to complete the E-step, we replace the sample version above with a stochastic approximation
\begin{align}
	Q_{\text{FC},n}^{(M)}(\bbeta, \btheta| \bX_{1:n}, \bY_{1:n} ; \hat{\bbeta}^{(t)}, \hat{\btheta}^{(t)}) &= \frac{1}{M} \sum_{m=1}^M \frac{1}{n} \sum_{i=1}^n \log p(\tilde{\bY}_i^{(m;t)} | \bX_i; \hat{\bbeta}^{(t)}, \hat{\btheta}^{(t)}) \nonumber,
\end{align}
where $\tilde{\bY}_i^{(m;t)} = (\bY_{i,\bR_i}, \tilde{\bY}_{i,\bar{\bR}_i}^{(m;t)})$.  To actually perform this stochastic approximation in practice, we simply multiply impute $M$ times and stack the imputed data sets together.

The M-step is equivalent to maximizing the latent incomplete likelihood on the stacked imputed data set $\tilde{D}_t = \{(\bX_i, \tilde{\bY}_i^{(m;t)})\}_{i=1,\ldots, n; \ m=1,\ldots,M}$.  Namely, we have

\begin{align*}
	(\hat{\bbeta}^{(t+1)}, \hat{\btheta}^{(t+1)}) = \argmax_{\bbeta, \btheta} \frac{1}{n} \sum_{(\bX_i, \tilde{\bY}_i^{(m;t)}) \in \tilde{D}_t} \ell_{\text{LI}}(\bbeta, \btheta ; \bX_i, \tilde{\bY}_i^{(m;t)}).
\end{align*}

We showed in Section 2 that this can be maximized via an EM algorithm (Algorithm \ref{alg:latentEM}).  Thus, the M-step of the outer EM algorithm is performed by Algorithm \ref{alg:latentEM}.  This completes the derivation of Algorithm \ref{alg:marEM}.
	
\subsection{Asymptotic Theory}

\begin{theorem}[Asymptotic distribution of the MLE]
	\label{thm:asymptotic}
Suppose that the following conditions hold:
\begin{enumerate}
	\item The parameter space is a compact set, and the true parameters lie within the interior.  That is, for some $\epsilon>0$ and $M>0$, $\theta_{0,k,j}\in[\epsilon, 1-\epsilon]$ and $\beta_{0,k,\ell}\in[-M,M]$ for all $k$, $j$, and $\ell$.  Here $\theta_{0,k,j}$ and $\beta_{0,k,j}$ denote the population-level MLE parameters.
	\item The parameters of the mixture of binomial product experts model are generically identified.
	\item The log-likelihood $(\bbeta,\btheta) \mapsto \log p(y|x;\bbeta,\btheta)$ is three times differentiable with respect to $(\bbeta, \btheta)$ for every $(x,y)$.
	\item The matrix $\E\left[\dfrac{\partial^2 \log p(Y|X;\bbeta,\btheta)}{\partial(\bbeta,\btheta)^2}\right]\biggr|_{\bbeta=\bbeta_0, \btheta=\btheta_0}$ exists and is nonsingular.
	\item The covariates $X$ are bounded.
	\item The missingness of the outcome variables satisfies a missing at random property.
\end{enumerate}
Then, there exists a unique population-level MLE $(\bbeta_0,\btheta_0)$ almost surely, and the maximum likelihood estimators $(\hat{\bbeta}, \hat{\btheta})$ obtained from Algorithm \ref{alg:marEM} satisfy
$$\sqrt{n}((\hat{\bbeta}, \hat{\btheta}) - (\bbeta_0,\btheta_0)) \stackrel{d}{\to} N(0, \Sigma)$$
for some positive definite $\Sigma$.
\end{theorem}

This theorem establishes mild sufficient conditions under which our MLE is asymptotically normal at a $\sqrt{n}$-rate.

Since the log-likelihood function is continuous over a compact set, the first assumption ensures that there exists a maximizer by the Extreme Value Theorem.  The second assumption ensures that under generic identification, the maximizer is unique almost surely.  The third and fourth assumptions ensure that we can apply a second order Taylor expansion (with a remainder term expressed in terms of the third order derivatives of the log-likelihood function) and obtain an asymptotically linear form.  The fifth assumption, combined with the fact that the outcome variables are test scores (bounded random variables), implies that the moments of the log-likelihood function exist and that higher order derivatives of the score function can be dominated over sufficiently small neighborhoods of $(\bbeta_0, \btheta_0)$.  This bounded assumption on the covariates is reasonable because the covariates are demographic variables.  Finally, the last assumption is on the missingness, and it allows the aforementioned MLE procedure to be valid.

Thus, the assumptions of Theorem 5.41 from \cite{van1998asymptotic} are satisfied, and we can apply that result here. It follows that Theorem \ref{thm:asymptotic} (asymptotic normality of the MLE) holds.

\section{Identification Theory}

\label{appendix:identification}


We provide sufficient conditions for generic identifiability in our mixture of binomial product experts model and restate the proposition here.

\setcounter{proposition}{0}
\begin{proposition}[Sufficient conditions for generic identifiability]
	Suppose the following conditions hold.
	\begin{enumerate}[label=A\arabic*]
		\item Each mixture is distinct such that $\theta_k \neq \theta_{k'}$ when $k \neq k'$. 
		\item The number of outcome variables $d$ and the number of mixtures $K$ satisfies the bound $d \geq 2\lceil \log_{1+\min_j N_j} K \rceil + 1$. 
		\item The design matrix is full-rank and $n> p$. 
	\end{enumerate}
	Then, the mixture of binomial product experts is generically identifiable up to permutation of the parameters.
\end{proposition}


Note that the bound in Assumption \ref{assump:inequality} is likely to be able to be relaxed even further because it was built off of a general latent class model, where each mixture is assumed to be a product of multinomials rather than binomials.  A mixture of binomial products is a submodel within the mixture of multinomial products model.  A similar bound is discussed in \cite{Allman2009} when all of the outcome variables have the same dimension.

The sufficient conditions for the identifiability of mixtures of binomials in the one dimensional case has been examined in \cite{Teicher1961}.  Assumption \ref{assump:identified-logistic} ensures that the design matrix is full column rank, which implies that the logistic regression parameters are identified since the the logit function is monotone.  More recently, \cite{ouyang2022identifiability} have established sufficient and necessary conditions for the identifiability of latent class models with covariates.

\section{Comments on Code Implementation}

\label{appendix:code-implement}

\subsection{Numerical Stability}

We use the logsumexp trick, which is commonly used for numerical stability \citep{statrethinkingbook}.  We are often interested in calculating probabilities of the form
$$(p_1,\ldots,p_K),$$
where $\sum_{k=1}^K p_k = 1$.  In our work, probabilities are often formed by the product of multiple terms in the range $[0,1]$.  Naive implementations can result in underflow, so we work with log-probabilities.  Additionally, consider the scenario where
$$p_k = \frac{\exp(t_k)}{\sum_{k'=1}^K \exp(t_{k'})}.$$
Then, we have
$$p_k = \exp\left(t_k - \left[c +\log \sum_{k'=1}^K \exp(t_{k'} - c) \right]\right)$$
for any constant $c$.  The form inside the brackets is known as the logsumexp trick.  If we pick $c = \max_{k'} t_{k'}$, then we can avoid most numerical underflows since each $t_{k'}$ is shifted.

\subsection{Random Initializations}

We randomly initialize $\bbeta$ and $\btheta$ using draws from uniform distributions.  The intercept of each $\beta_k$ is drawn from $U[-1,1]$ while each slope term is drawn from $U[-0.5,0.5]$.  We draw each $\theta_{k,j}$ from $U[0.1,0.9]$, so that it is initialized away from the boundary.  All draws are performed independently.

\section{Remarks on Alternative EM Algorithms}

\label{appendix:alternativeEM}

\begin{remark}
	
	\label{subsect:alternativeEM}
	
	For completeness, we provide some remarks on an alternative EM algorithm for fitting the model.  Our proposed method comprises an EM algorithm nested within an outer EM algorithm.  Alternatively, one can impute the missing outcome and latent class variables simultaneously using the distributions $P(Z=k|x,y_r; \hat{\bbeta}^{(t)}, \hat{\btheta}^{(t)})$ and $p(y_{\bar{r}}| Z=k, y_r, x; \hat{\bbeta}^{(t)}, \hat{\btheta}^{(t)})$ and update $(\hat{\bbeta}^{(t+1)}, \hat{\btheta}^{(t+1)})$ with the explicit MLE estimates.  When there is no missingness, this procedure mimics the Stochastic EM procedure described by \cite{celeuxd85} and \cite{CeleuxSEM1987}.  A similar method has been explored in the missing at random setting under a Gaussian mixture model with no covariates by \cite{serafini2020handling}.  Alternating between these two steps until convergence provides a consistent estimate.  However, we do not recommend this method for our model in practice because time to convergence is often longer than Algorithm \ref{alg:marEM} for the same number of imputations $M$.  We anticipate this is due to the increased stochastic variation that is observed by imputing $Z$ on each iteration.
	
\end{remark}

\begin{remark}
	
	There has been previous work on model-based clustering with missing data.  \cite{serafini2020handling} discussed how to estimate Gaussian mixture models in the presence of missing at random data with the EM algorithm and Monte Carlo methods in the E-step.  We provide comments on how our approach compares to a similar approach in Remark \ref{subsect:alternativeEM}.  Our approach differs because we utilize two EM algorithms, one nested within the other.  Unlike some of the previous work in the model-based clustering with missing data literature, we take a more broad approach by incorporating covariates in our model as well as describing an inference procedure that quantifies uncertainty.  \cite{sportisse2023modelbased} described model-based clustering with missing not at random assumptions using a likelihood-based approach and an EM algorithm.  Although there has been rich literature on missing not at random assumptions, there are still open problems with the missing at random assumption (such as performing sensitivity analysis) and sometimes missing not at random assumptions are not easily interpretable.
\end{remark}

\section{Simulations}

\label{sect:simulations}

\subsection{Simulation Details}

\label{appendix:simdetails}


We describe the data generating process.  We have $p=2$ covariates and $d=4$ outcome variables.  The model is described as follows:
\begin{align*}
	p(x) = N\left( \begin{bmatrix} 2 \\ 3 \end{bmatrix}, \begin{bmatrix} 1 & 0.2 \\ 0.2 & 1 \end{bmatrix} \right), \quad p(y|x) = \sum_{k=1}^3 w_k(x;\beta_k) \prod_{j=1}^4 \binom{N_j}{y_j} \theta_{k,j}^{y_j} (1-\theta_{k,j})^{N_j-y_j},
\end{align*}
where $\beta_1 = (0,0,0)$, $\beta_2 = (-1.5, 0.3, 0.4)$, and $\beta_3 = (-2, 0.5, 0.25)$.  The parameters for the binomials are as follows
\begin{equation*}
	\btheta = \begin{bmatrix} \theta_1 \\ \theta_2 \\ \theta_3 \end{bmatrix} = \begin{bmatrix} 0.8 & 0.8 & 0.8 & 0.8 \\ 0.5 & 0.5 & 0.5 & 0.5 \\ 0.1 & 0.1 & 0.1 & 0.1 \end{bmatrix}.
\end{equation*}

Thus, for a moderately sized model, we already have 18 parameters.  To generate data for the missing data simulations, we use the same generating process as above, but we specify a selection model $P(R=r|x,y)$ to make the missing data.  For every missing data pattern $r\neq 1111$, we generate the probabilities $P(R=r|x,y) = P(R=r|x,y_r)$ under the following scheme to ensure MAR holds.  We have
\begin{align*}
	&\log\frac{P(R=0001|x,y)}{P(R\neq0001|x,y)} = (-2-\eta) - 0.25 x_1 + 0.3 x_2 + 0.15 y_4,\\
	&\log\frac{P(R=0110|x,y)}{P(R\neq0110|x,y)} = (-1-\eta) + 0.3 x_1 - 0.7 x_2 - 0.1 y_2 + 0.15 y_3,\\
	&\log\frac{P(R=1010|x,y)}{P(R\neq1010|x,y)} = (-2-\eta) + 0.7 x_1 - 0.4 x_2 + 0.24 y_1 - 0.15 y_3,\\
	&\log\frac{P(R=1110|x,y)}{P(R\neq1110|x,y)} = (-1-\eta) + 0.2 x_1 - 0.15 x_2 + 0.15 y_1 - 0.14 y_2 + 0.05 y_3, \\
	&P(R=1111) = 1 - \sum_{r\neq 1111} P(R=r|x,y).
\end{align*}

Here we treat $\eta \in [0,\infty)$ as a a parameter that controls the amount of missingness.  As $\eta\to\infty$, $P(R=r|x,y) \to 0$ for every $r\neq 1111$, thereby, leading towards $P(R=1111|x,y)\to 1$ and consequently, a data set with no missing outcome variables.

\subsection{Additional Simulations}

	\label{appendix:add-sim}
	
	We provide additional simulations to examine the effect of the number of imputations on the results of the Monte Carlo EM algorithm.  We repeat the experiment for $\eta=2,2.5$ with $M=10$ and $M=50$ imputations.  We do not repeat it for $\eta=\infty$ because the latter has no missing data and thus, requires no imputation.  In this particular simulation, we actually observe similar performance as we vary from $M=10,20,50$ imputations, suggesting that already at $M=10$ imputations, the Monte Carlo error is sufficiently small.  Additionally, since we are able to know the true generating mixture component for each data point, we compare the assigned clustering to the true generating mixture component.  We report the average adjusted Rand index for $M=10,20,50$ in Tables \ref{table:M10_sim}, \ref{table:M20_randindex}, and \ref{table:M50_sim}.  We observe fairly high and stable adjusted Rand indices across all sample sizes, amount of missingness $\eta$, and number of imputations $M$.  This agrees with our expectations of our estimation procedure since the mixtures are well-separated in the generating process.
	
	Furthermore, we run a set of simulations to determine the performance of a complete case analysis.  The results are summarized in Table \ref{table:CC_sim}.  We observe that the MSE does not decrease at a linear rate due to the bias associated with not accounting for the missing data and simply performing a complete case analysis.  Thus, the estimates are not consistent.  The adjusted Rand index is still mostly high with comparable performance to the estimation procedure that takes into account missingness.  This can likely be attributed to the fact that the simulation problem has well-separated clusters.  In our work, we are also interested in statistical inference and confidence intervals for estimating parameters.  As sample size increases, the Estimated Coverage for the parameters generally decreases, away from the desired 95\% coverage, implying that the procedure does not produce asymptotically valid confidence intervals.  Thus, this simulation highlights the necessity of taking into account the missingness when constructing an estimation procedure and how complete case analysis can yield statistically invalid results.

	\begin{table}
		
		\caption{These are the results for $M=10$ imputations.  The first row contains results for the estimation of $\btheta$.  The second row contains results for the estimation of $\bbeta$.  The third row contains the average adjusted Rand index for each combination of $\eta$ and sample size $n$.}
		\label{table:M10_sim}
		\centering
		\scalebox{1}{
			
			\begin{tabular}{cccc|}
				\multicolumn{4}{c}{$\textbf{MSE}_{\btheta} (\times 100)$}                                                                                                               \\
				\cline{2-4}
				\multicolumn{1}{c|}{}        & \multicolumn{3}{c|}{Sample Size}                                                                        \\
				\cline{1-4}
				\multicolumn{1}{|c|}{}        & \multicolumn{1}{c|}{$n=500$} & \multicolumn{1}{c|}{$n=1000$} & \multicolumn{1}{c|}{$n=2000$} \\
				\hline
				\multicolumn{1}{|c|}{$\eta=2.5$} & \multicolumn{1}{c|}{0.190}        & \multicolumn{1}{c|}{0.094}         & \multicolumn{1}{c|}{0.049}         \\
				\multicolumn{1}{|c|}{$\eta=2$} & \multicolumn{1}{c|}{0.202}        & \multicolumn{1}{c|}{0.101}         & \multicolumn{1}{c|}{0.054}              
			\end{tabular}
			
			\hspace{10pt}
			
			\begin{tabular}{cccc|}
				\multicolumn{4}{c}{\textbf{Estimated Coverage of} $\theta$s}                                                                                                               \\
				\cline{2-4}
				\multicolumn{1}{c|}{}        & \multicolumn{3}{c|}{Sample Size}                                                                        \\
				\cline{1-4}
				\multicolumn{1}{|c|}{}        & \multicolumn{1}{c|}{$n=500$} & \multicolumn{1}{c|}{$n=1000$} & \multicolumn{1}{c|}{$n=2000$} \\
				\hline
				\multicolumn{1}{|c|}{$\eta=2.5$} & \multicolumn{1}{c|}{0.940}        & \multicolumn{1}{c|}{0.940}         & \multicolumn{1}{c|}{0.937}         \\
				\multicolumn{1}{|c|}{$\eta=2$} & \multicolumn{1}{c|}{0.939}        & \multicolumn{1}{c|}{0.938}         & \multicolumn{1}{c|}{0.931}           \\
			\end{tabular}

		}
		
		\vskip 10pt
		
		\scalebox{1}{
			\begin{tabular}{cccc|}
				\multicolumn{4}{c}{$\textbf{MSE}_{\bbeta} (\times 100)$}                                                                                                               \\
				\cline{2-4}
				\multicolumn{1}{c|}{}        & \multicolumn{3}{c|}{Sample Size}                                                                        \\
				\cline{1-4}
				\multicolumn{1}{|c|}{}        & \multicolumn{1}{c|}{$n=500$} & \multicolumn{1}{c|}{$n=1000$} & \multicolumn{1}{c|}{$n=2000$} \\
				\hline
				\multicolumn{1}{|c|}{$\eta=2.5$} & \multicolumn{1}{c|}{45.5}        & \multicolumn{1}{c|}{21.3}         & \multicolumn{1}{c|}{10.3}         \\
				\multicolumn{1}{|c|}{$\eta=2$} & \multicolumn{1}{c|}{43.8}        & \multicolumn{1}{c|}{21.4}         & \multicolumn{1}{c|}{10.5}              
			\end{tabular}
			
			\hspace{10pt}
			
			\begin{tabular}{cccc|}
				\multicolumn{4}{c}{\textbf{Estimated Coverage of} $\beta$s}                                                                                                               \\
				\cline{2-4}
				\multicolumn{1}{c|}{}        & \multicolumn{3}{c|}{Sample Size}                                                                        \\
				\cline{1-4}
				\multicolumn{1}{|c|}{}        & \multicolumn{1}{c|}{$n=500$} & \multicolumn{1}{c|}{$n=1000$} & \multicolumn{1}{c|}{$n=2000$} \\
				\hline
				\multicolumn{1}{|c|}{$\eta=2.5$} & \multicolumn{1}{c|}{0.953}        & \multicolumn{1}{c|}{0.952}         & \multicolumn{1}{c|}{0.953}         \\
				\multicolumn{1}{|c|}{$\eta=2$} & \multicolumn{1}{c|}{0.956}        & \multicolumn{1}{c|}{0.950}         & \multicolumn{1}{c|}{0.949}              
			\end{tabular}
			
		}
	
		\scalebox{1}{
		
		\begin{tabular}{cccc|}
			\multicolumn{4}{c}{Average Adjusted Rand Index}                                                                                                               \\
			\cline{2-4}
			\multicolumn{1}{c|}{}        & \multicolumn{3}{c|}{Sample Size}                                                                        \\
			\cline{1-4}
			\multicolumn{1}{|c|}{}        & \multicolumn{1}{c|}{$n=500$} & \multicolumn{1}{c|}{$n=1000$} & \multicolumn{1}{c|}{$n=2000$} \\
			\hline
			\multicolumn{1}{|c|}{$\eta=2.5$} & \multicolumn{1}{c|}{0.935}        & \multicolumn{1}{c|}{0.936}         & \multicolumn{1}{c|}{0.937}         \\
			\multicolumn{1}{|c|}{$\eta=2$} & \multicolumn{1}{c|}{0.935}        & \multicolumn{1}{c|}{0.936}         & \multicolumn{1}{c|}{0.936}              
		\end{tabular}
		
	}
		
	\end{table}

	\begin{table}
	\caption{This table reports the average adjusted Rand index for each combination of $\eta$ and sample size $n$ for $M=20$ imputations over all $1000$ generated data sets.}
	\label{table:M20_randindex}
	\centering
	\scalebox{1}{
		
	\begin{tabular}{cccc|}
		\multicolumn{4}{c}{Average Adjusted Rand Index}                                                                                                               \\
		\cline{2-4}
		\multicolumn{1}{c|}{}        & \multicolumn{3}{c|}{Sample Size}                                                                        \\
		\cline{1-4}
		\multicolumn{1}{|c|}{}        & \multicolumn{1}{c|}{$n=500$} & \multicolumn{1}{c|}{$n=1000$} & \multicolumn{1}{c|}{$n=2000$} \\
		\hline
		\multicolumn{1}{|c|}{$\eta=2.5$} & \multicolumn{1}{c|}{0.936}        & \multicolumn{1}{c|}{0.936}         & \multicolumn{1}{c|}{0.937}         \\
		\multicolumn{1}{|c|}{$\eta=2$} & \multicolumn{1}{c|}{0.934}        & \multicolumn{1}{c|}{0.936}         & \multicolumn{1}{c|}{0.936}              
	\end{tabular}
	
	}
	
	\end{table}
	
	\begin{table}
		\caption{These are the results for $M=50$ imputations.  The first row contains results for the estimation of $\btheta$.  The second row contains results for the estimation of $\bbeta$.  The third row contains the average adjusted Rand index for each combination of $\eta$ and sample size $n$.}
		\label{table:M50_sim}
		\centering
		\scalebox{1}{
			
			\begin{tabular}{cccc|}
				\multicolumn{4}{c}{$\textbf{MSE}_{\btheta} (\times 100)$}                                                                                                               \\
				\cline{2-4}
				\multicolumn{1}{c|}{}        & \multicolumn{3}{c|}{Sample Size}                                                                        \\
				\cline{1-4}
				\multicolumn{1}{|c|}{}        & \multicolumn{1}{c|}{$n=500$} & \multicolumn{1}{c|}{$n=1000$} & \multicolumn{1}{c|}{$n=2000$} \\
				\hline
				\multicolumn{1}{|c|}{$\eta=2.5$} & \multicolumn{1}{c|}{0.189}        & \multicolumn{1}{c|}{0.094}         & \multicolumn{1}{c|}{0.049}         \\
				\multicolumn{1}{|c|}{$\eta=2$} & \multicolumn{1}{c|}{0.204}        & \multicolumn{1}{c|}{0.104}         & \multicolumn{1}{c|}{0.054}              
			\end{tabular}
			
			\hspace{10pt}
			
			\begin{tabular}{cccc|}
				\multicolumn{4}{c}{\textbf{Estimated Coverage of} $\theta$s}                                                                                                               \\
				\cline{2-4}
				\multicolumn{1}{c|}{}        & \multicolumn{3}{c|}{Sample Size}                                                                        \\
				\cline{1-4}
				\multicolumn{1}{|c|}{}        & \multicolumn{1}{c|}{$n=500$} & \multicolumn{1}{c|}{$n=1000$} & \multicolumn{1}{c|}{$n=2000$} \\
				\hline
				\multicolumn{1}{|c|}{$\eta=2.5$} & \multicolumn{1}{c|}{0.942}        & \multicolumn{1}{c|}{0.941}         & \multicolumn{1}{c|}{0.940}         \\
				\multicolumn{1}{|c|}{$\eta=2$} & \multicolumn{1}{c|}{0.938}        & \multicolumn{1}{c|}{0.931}         & \multicolumn{1}{c|}{0.932}           \\
			\end{tabular}

		}
		
		\vskip 10pt
		
		\scalebox{1}{
			\begin{tabular}{cccc|}
				\multicolumn{4}{c}{$\textbf{MSE}_{\bbeta} (\times 100)$}                                                                                                               \\
				\cline{2-4}
				\multicolumn{1}{c|}{}        & \multicolumn{3}{c|}{Sample Size}                                                                        \\
				\cline{1-4}
				\multicolumn{1}{|c|}{}        & \multicolumn{1}{c|}{$n=500$} & \multicolumn{1}{c|}{$n=1000$} & \multicolumn{1}{c|}{$n=2000$} \\
				\hline
				\multicolumn{1}{|c|}{$\eta=2.5$} & \multicolumn{1}{c|}{44.6}        & \multicolumn{1}{c|}{21.8}         & \multicolumn{1}{c|}{10.7}         \\
				\multicolumn{1}{|c|}{$\eta=2$} & \multicolumn{1}{c|}{47.1}        & \multicolumn{1}{c|}{22.0}         & \multicolumn{1}{c|}{10.4}              
			\end{tabular}
			
			\hspace{10pt}
			
			\begin{tabular}{cccc|}
				\multicolumn{4}{c}{\textbf{Estimated Coverage of} $\beta$s}                                                                                                               \\
				\cline{2-4}
				\multicolumn{1}{c|}{}        & \multicolumn{3}{c|}{Sample Size}                                                                        \\
				\cline{1-4}
				\multicolumn{1}{|c|}{}        & \multicolumn{1}{c|}{$n=500$} & \multicolumn{1}{c|}{$n=1000$} & \multicolumn{1}{c|}{$n=2000$} \\
				\hline
				\multicolumn{1}{|c|}{$\eta=2.5$} & \multicolumn{1}{c|}{0.956}        & \multicolumn{1}{c|}{0.947}         & \multicolumn{1}{c|}{0.951}         \\
				\multicolumn{1}{|c|}{$\eta=2$} & \multicolumn{1}{c|}{0.947}        & \multicolumn{1}{c|}{0.947}         & \multicolumn{1}{c|}{0.951}              
			\end{tabular}
		}
		\vskip 10pt

		\begin{tabular}{cccc|}
			\multicolumn{4}{c}{Average Adjusted Rand Index}                                                                                                               \\
			\cline{2-4}
			\multicolumn{1}{c|}{}        & \multicolumn{3}{c|}{Sample Size}                                                                        \\
			\cline{1-4}
			\multicolumn{1}{|c|}{}        & \multicolumn{1}{c|}{$n=500$} & \multicolumn{1}{c|}{$n=1000$} & \multicolumn{1}{c|}{$n=2000$} \\
			\hline
			\multicolumn{1}{|c|}{$\eta=2.5$} & \multicolumn{1}{c|}{0.935}        & \multicolumn{1}{c|}{0.936}         & \multicolumn{1}{c|}{0.937}         \\
			\multicolumn{1}{|c|}{$\eta=2$} & \multicolumn{1}{c|}{0.936}        & \multicolumn{1}{c|}{0.936}         & \multicolumn{1}{c|}{0.936}              
		\end{tabular}

	\end{table}
	
	\begin{table}
		\caption{These are the results for complete case analysis.  The first row contains results for the estimation of $\btheta$.  The second row contains results for the estimation of $\bbeta$.}
		\label{table:CC_sim}
		\centering
		\scalebox{0.75}{
			
			\begin{tabular}{cccccc|}
				\multicolumn{6}{c}{$\textbf{MSE}_{\btheta} (\times 100)$}                                                                                                               \\
				\cline{2-6}
				\multicolumn{1}{c|}{}        & \multicolumn{5}{c|}{Sample Size}                                                                        \\
				\cline{1-6}
				\multicolumn{1}{|c|}{}        & \multicolumn{1}{c|}{$n=500$} & \multicolumn{1}{c|}{$n=1000$} & \multicolumn{1}{c|}{$n=2000$} &
				\multicolumn{1}{c|}{$n=4000$} &
				\multicolumn{1}{c|}{$n=8000$} \\
				\hline
				\multicolumn{1}{|c|}{$\eta=2.5$} & \multicolumn{1}{c|}{0.196}        & \multicolumn{1}{c|}{0.101}         & \multicolumn{1}{c|}{0.050}         &
				\multicolumn{1}{c|}{0.028} &
				\multicolumn{1}{c|}{0.015}         \\
				\multicolumn{1}{|c|}{$\eta=2$} & \multicolumn{1}{c|}{0.220}        & \multicolumn{1}{c|}{0.115}         & \multicolumn{1}{c|}{0.061}          & 
				\multicolumn{1}{c|}{0.033} &
				\multicolumn{1}{c|}{0.022}  
			\end{tabular}
			
			\hspace{10pt}
			
			\begin{tabular}{cccccc|}
				\multicolumn{6}{c}{\textbf{Estimated Coverage of} $\theta$s}                                                                                                               \\
				\cline{2-6}
				\multicolumn{1}{c|}{}        & \multicolumn{5}{c|}{Sample Size}                                                                        \\
				\cline{1-6}
				\multicolumn{1}{|c|}{}        & \multicolumn{1}{c|}{$n=500$} & \multicolumn{1}{c|}{$n=1000$} & \multicolumn{1}{c|}{$n=2000$} &
				\multicolumn{1}{c|}{$n=4000$} &
				\multicolumn{1}{c|}{$n=8000$} \\
				\hline
				\multicolumn{1}{|c|}{$\eta=2.5$} & \multicolumn{1}{c|}{0.943}        & \multicolumn{1}{c|}{0.936}         &
				\multicolumn{1}{c|}{0.936} &
				\multicolumn{1}{c|}{0.914} & \multicolumn{1}{c|}{0.887}         \\
				\multicolumn{1}{|c|}{$\eta=2$} & \multicolumn{1}{c|}{0.938}        & \multicolumn{1}{c|}{0.939}         &
				\multicolumn{1}{c|}{0.913} &
				\multicolumn{1}{c|}{0.887} & \multicolumn{1}{c|}{0.843}           \\
			\end{tabular}

		}
		
		\vskip 10pt

		\scalebox{0.75}{
			\begin{tabular}{cccccc|}
				\multicolumn{6}{c}{$\textbf{MSE}_{\bbeta} (\times 100)$}                                                                                                               \\
				\cline{2-6}
				\multicolumn{1}{c|}{}        & \multicolumn{5}{c|}{Sample Size}                                                                        \\
				\cline{1-6}
				\multicolumn{1}{|c|}{}        & \multicolumn{1}{c|}{$n=500$} & \multicolumn{1}{c|}{$n=1000$} &
				\multicolumn{1}{c|}{$n=2000$} & \multicolumn{1}{c|}{$n=4000$} &
				\multicolumn{1}{c|}{$n=8000$} \\
				\hline
				\multicolumn{1}{|c|}{$\eta=2.5$} & \multicolumn{1}{c|}{49.6}        & \multicolumn{1}{c|}{25.1}         &
				\multicolumn{1}{c|}{13.4} &
				\multicolumn{1}{c|}{7.66} & \multicolumn{1}{c|}{4.80}         \\
				\multicolumn{1}{|c|}{$\eta=2$} & \multicolumn{1}{c|}{59.0}        & \multicolumn{1}{c|}{28.3}         &
				\multicolumn{1}{c|}{17.1} &
				\multicolumn{1}{c|}{10.7}     & \multicolumn{1}{c|}{7.71}              
			\end{tabular}
			
			\hspace{10pt}
			
			\begin{tabular}{cccccc|}
				\multicolumn{6}{c}{\textbf{Estimated Coverage of} $\beta$s}                                                                                                               \\
				\cline{2-6}
				\multicolumn{1}{c|}{}        & \multicolumn{5}{c|}{Sample Size}                                                                        \\
				\cline{1-6}
				\multicolumn{1}{|c|}{}        & \multicolumn{1}{c|}{$n=500$} & \multicolumn{1}{c|}{$n=1000$} &
				\multicolumn{1}{c|}{$n=2000$} &
				\multicolumn{1}{c|}{$n=4000$} & \multicolumn{1}{c|}{$n=8000$} \\
				\hline
				\multicolumn{1}{|c|}{$\eta=2.5$} & \multicolumn{1}{c|}{0.940}        & \multicolumn{1}{c|}{0.945}         &
				\multicolumn{1}{c|}{0.935} &
				\multicolumn{1}{c|}{0.904} & \multicolumn{1}{c|}{0.843}         \\
				\multicolumn{1}{|c|}{$\eta=2$} & \multicolumn{1}{c|}{0.947}        & \multicolumn{1}{c|}{0.939}         &
				\multicolumn{1}{c|}{0.908} & \multicolumn{1}{c|}{0.871} &
				\multicolumn{1}{c|}{0.789} 
			\end{tabular}
			
		}
	
		\scalebox{0.75}{
			
			\begin{tabular}{cccccc|}
				\multicolumn{6}{c}{Average Adjusted Rand Index}                                                                                                               \\
				\cline{2-6}
				\multicolumn{1}{c|}{}        & \multicolumn{5}{c|}{Sample Size}                                                                        \\
				\cline{1-6}
				\multicolumn{1}{|c|}{}        & \multicolumn{1}{c|}{$n=500$} & \multicolumn{1}{c|}{$n=1000$} & \multicolumn{1}{c|}{$n=2000$} &
				\multicolumn{1}{c|}{$n=4000$} &
				\multicolumn{1}{c|}{$n=8000$} \\
				\hline
				\multicolumn{1}{|c|}{$\eta=2.5$} & \multicolumn{1}{c|}{0.936}        & \multicolumn{1}{c|}{0.936}         & \multicolumn{1}{c|}{0.937}         &
				\multicolumn{1}{c|}{0.937} &
				\multicolumn{1}{c|}{0.937}         \\
				\multicolumn{1}{|c|}{$\eta=2$} & \multicolumn{1}{c|}{0.936}        & \multicolumn{1}{c|}{0.936}         & \multicolumn{1}{c|}{0.936}          & 
				\multicolumn{1}{c|}{0.937} &
				\multicolumn{1}{c|}{0.937}  
			\end{tabular}
		}
	\end{table}

\section{Comments on Model Selection}

\label{appendix:modelselect}

\subsection{Model selection}

In a standard data analysis, one key question is choosing the number of mixture components.  In the case of no prior knowledge, we recommend using information criteria for model selection.  The Akaike information criterion (AIC) \cite{akaike1973information} and the Bayesian information criterion (BIC) \cite{Schwarz1978} are popular methods for choosing the number of mixture components.  Each of these information criteria works by adding a penalty term to the observed log-likelihood, defined above in \eqref{eq:obs-loglike}.

Let $\nu(K)$ be the number of free parameters of the mixture of binomial product experts model with $K$ mixtures.  We showed earlier that $\nu(K) = K(p+d+1)-p-1$.  As functions of $K$, the AIC and BIC write as
\begin{align*}
	\quad \text{AIC}(K) := 2\nu(K) - 2\ell_{\text{obs},n}, \quad \text{BIC}(K) := \nu(K)\log n - 2 \ell_{\text{obs},n},
\end{align*}
respectively.  One can choose the model that minimizes the AIC or BIC.

In another simulation, we repeatedly generate a data set of size $n=500$ with $\eta=\infty$ and $\eta=2$ (details in Appendix \ref{sect:simulations}).  Then, we fit models of varying sizes with $K=2,3,4,5,6$ and create the AIC and BIC curves.  For each data set, we find the numbers of mixtures that corresponds to the minimum AIC and BIC.  Figure~\ref{fig:AICBIC} provides an example of AIC and BIC curves for varying values of $K$ for two specific data sets.
	
	We generate 100 random data sets of size $n=500$ with $\eta=2,\infty$, and determine the number of mixture components selected by the AIC and BIC for each data set.  The results are summarized in Table \ref{table:AICBIC_repeat}.  We see that in general, the AIC is not as reliable as the BIC.  The BIC has a higher penalty term than the AIC, and it chooses the recommended model for all data sets in the simulations.  Therefore, this simulation suggests using BIC over AIC as an information criterion.  We also encourage the use of prior scientific knowledge when possible for selecting a model.

\begin{table}
	\caption{These are the results of generating 100 data sets of size $n=500$ with $\eta=2,\infty$ and computing the AIC and BIC for each.  We record for each of the 100 generated sets, which number of mixture of components the AIC and the BIC chose.  The first table is for $\eta=\infty, n=500$, and the second table is for $\eta=2, n=500$.}
	\label{table:AICBIC_repeat}
	\centering
	\vskip 10pt
	\scalebox{1}{
	\begin{tabular}{c|ccccc}
				$\eta = \infty, n= 500$ & $K=2$ & $K=3$ & $K=4$ & $K=5$ & $K=6$ \\
				\hline
				AIC                     & 0     & 22    & 51    & 14    & 13    \\
				BIC                     & 0     & 100   & 0     & 0     & 0    
			\end{tabular}
	}
	\vskip 10pt
	\scalebox{1}{
	\begin{tabular}{c|ccccc}
		$\eta = 2, n= 500$ & $K=2$ & $K=3$ & $K=4$ & $K=5$ & $K=6$ \\
		\hline
		AIC                     & 0     & 47    & 31    & 13    & 9    \\
		BIC                     & 0     & 100   & 0     & 0     & 0    
	\end{tabular}
}

\end{table}

\begin{figure}
	\centering
	\includegraphics[width=6.75cm]{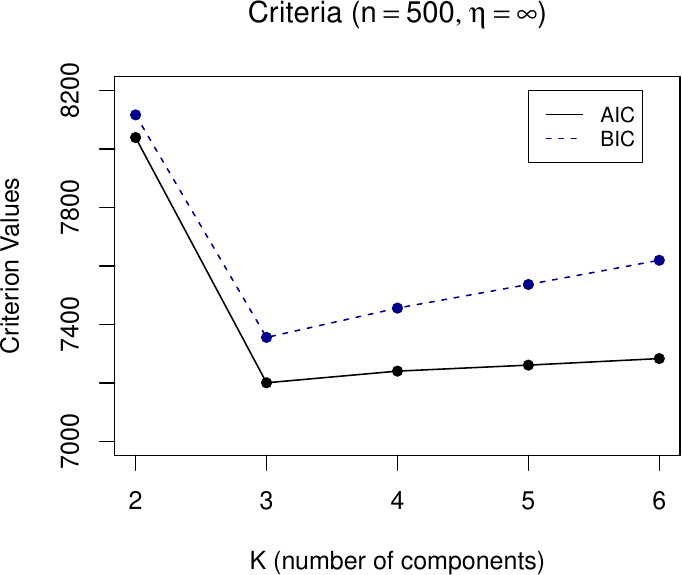}
	\hfill
	\includegraphics[width=6.75cm]{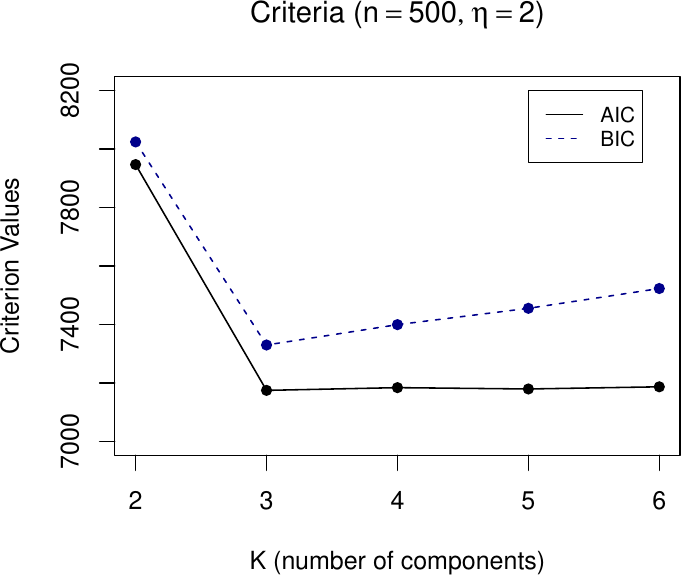}
	\caption{This figure depicts the AIC and BIC curves as we vary $K$ for a given simulated data set of size $n=500$ and $\eta=\infty,2$.}
	\label{fig:AICBIC}
\end{figure}

%
%
%
%
%
%
%

\section{Additional Comments on Clustering}

\label{appendix:clustering}

Interestingly, the clustering procedure in the presence of missing data that is outlined in Section \ref{sect:clustering} can be related to a multiple imputation based approach as follows.  Consider imputing $M$ data sets and computing the following mean
	$$\hat{\pi}_{k,\bR_i}^{(M)}(\bX_i, \bY_{i,\bR_i}):=\frac{1}{M}\sum_{m=1}^M \hat{\pi}_{k}(\bX_i, \tilde{\bY}_i^{(m)})$$
	for every observation $i$.  This procedure can be thought of constructing $M$ plausible complete data sets and averaging the latent class probability vectors for each observation $i$ over all of the $M$ imputed data sets.  Then, one can perform clustering using the probability $\hat{\pi}_{k,\bR_i}^{(M)}(\bX_i, \bY_{i,\bR_i})$.  As $M\to\infty$, we expect $\hat{\pi}_{k,\bR_i}^{(M)}(\bX_i, \bY_{i,\bR_i}) \stackrel{p}{\to} \hat{\pi}_{k,\bR_i}(\bX_i, \bY_{i,\bR_i})$ because for any $r\in\mathcal{R}$,
	$$\E[P(Z=k|\bX,\bY) | \bX, \bY_r] = P(Z=k|\bX,\bY_r).$$
	
	Thus, for sufficiently large $M$, the two procedures are equivalent because higher $M$ results in more imputations, which reduces the Monte Carlo error that arises from the variability of the imputation.  However, since there is an explicit form for the probability $\hat{\pi}_{k,r}(\bX,\bY)$, there is no need to do multiple imputation.
	
	\begin{remark}
		We use the notation $\hat{\pi}$ to emphasize that these probabilities are based on an \textit{estimated model from a sample}.  We estimate these probabilities by plugging our estimators into our statistical model and solving for the desired probabilities with Bayes' rule.  In theory, there exist oracle probabilities $\pi_k(\bX,\bY) := P(Z=k | \bX, \bY)$ and $\pi_{k,r}(\bX,\bY_r) := P(Z=k | \bX, \bY_r)$ based on the population.  One can analyze clustering based on these population quantities, but in practice, we do not have access to these functions because they are parameterized by the true population parameters.  Instead, we use $\hat{\pi}_k(\bX,\bY)$ and $\hat{\pi}_{k,r}(\bX, \bY_r)$ as approximations to the population quantities.  More theoretical analysis can be performed to compare the clustering of our sample to a population clustering, but this is out of the scope of this paper.
	\end{remark}

\section{Proofs}

\label{appendix:proofs}

\subsection{Nonconcavity of the Latent Incomplete Log-Likelihood (Lemma \ref{lemma:nonconcavity})}

\label{appendix:lemmaproof}

We first demonstrate that the latent incomplete log-likelihood function is nonconcave.

	\begin{lemma}[Nonconcavity of the LI log-likelihood function]
		\label{lemma:nonconcavity}
		The latent incomplete log-likelihood function
		\begin{align*}
			\ell_{\text{LI},n}(\beta,\theta; \bX_{1:n}, \bY_{1:n}) &= \log \left( \prod_{i=1}^n \left( \sum_{k=1}^K w_k(\bX_i; \beta_k) p_k(\bY_i ; \theta_k)\right) \right)\\
			&= \sum_{i=1}^n \log \left( \sum_{k=1}^K w_k(\bX_i; \beta_k) p_k(\bY_i ; \theta_k)\right).
		\end{align*}
		is not concave.
	\end{lemma}

\begin{proof}[Proof of Lemma \ref{lemma:nonconcavity}]
	We consider a simple counterexample showing nonconcavity.  Take $n=1$, $K=2$, $p=1$, and $d=3$.  We further assume that $N_1=N_2=N_3=5$.  This is the log-likelihood where there is exactly one data point with two clusters, a single covariate, and three bounded discrete outcome variables.  Suppose that $\bX=1$ and $\bY=(3,1,4)$.  We have
	\begin{align*}
		\ell_{\text{LI}, n}(\bbeta, \btheta | \bX, \bY) = \log\left( \sum_{k=1}^2 w_k(X; \beta_k) \prod_{j=1}^3 \binom{N_j}{y_j} \theta_{k,j}^{N_j} (1-\theta_{k,j})^{N_j-y_j} \right).
	\end{align*}
	
	Consider two points $(\bbeta^1, \btheta^1)$ and $(\bbeta^2, \btheta^2)$, where the superscript indexes the data point.  Let $\bbeta^1 = (-1,-1)$, $\btheta^1 = (\theta_1^1, \theta_2^1)$, $\theta_1^1 = (0.7,0.6,0.3)$, $\theta_2^1 = (0.3,0.5,0.1)$, $\bbeta^2 = (1,6)$, $\btheta^2 = (\theta_1^2, \theta_2^2)$, $\theta_1^2 = (0.6,0.4,0.8)$, and $\theta_2^2 = (0.4,0.3,0.3)$.  We now consider the line that connects these two points and show that it sits above the function.
	
	Observe that
	$$\ell_{\text{LI}, n}(t\bbeta^1 + (1-t)\bbeta^2, t\btheta^1 + (1-t)\btheta^2) \approx -6.81 < -6.73 \approx t\ell_{\text{LI}, n}(\bbeta^1, \btheta^1) + (1-t) \ell_{\text{LI}, n}(\bbeta^2, \btheta^2)$$
	for $t=1/2$.  Thus, nonconcavity is achieved.
\end{proof}

\subsection{Identification Proposition (Proposition \ref{prop:identif})}

Before we present the proof of this proposition, We first introduce a result from \cite{Allman2009} on generic identification.  Then, we state two helpful lemmas and providing their proofs.

\begin{theorem}[Theorem 4 of \cite{Allman2009}]
	Let $S_1,\ldots, S_d$ be $d$ categorical variables
	and $S_j$ has $B_j$ categories. 
	Consider a $K$-mixture model: 
	$$
	P(S=s; \lambda) = {\displaystyle\sum_{k=1}^K} \omega_k \prod_{j=1}^d q(s_{k,j}; \lambda_{k,j}),
	$$
	where $q(s_{k,j}; \lambda_{k,j})$ is a multinomial distribution on the $B_j$ categories. 
	Suppose there exists a partition $C_1,C_2,C_3$ of the set $\{1,2,\cdots, d\}$ such that 
	$b_j = \prod_{j\in C_j} B_j$ and 
	$$
	\min\{K, b_1\}+\min\{K, b_2\}+\min\{K, b_3\}\geq 2K +2. 
	$$
	Then, the model $P(S=s; \lambda)$ is generically identifiable.
	\label{thm::allman}
\end{theorem}

The first lemma describes the generic identifiability of the mixture of binomial products model without covariates.

\begin{lemma}
	Consider the model
	$$\sum_{k=1}^K w_k \left(\prod_{j=1}^d \binom{N_j}{y_j} (\theta_{k,j})^{y_j} (1-\theta_{k,j})^{N_j-y_j}\right).$$
	If the dimension of the outcome variable $Y$ and the number of components satisfy $d\geq 2\lceil \log_{1+\min_j N_j} K \rceil + 1$, then the model is generically identified up to permutation of the parameters.
	\label{lemma:nocovariate-identification}
\end{lemma}

\begin{proof}[Proof of Lemma \ref{lemma:nocovariate-identification}]
	This result is a consequence of Theorem 4 from \cite{Allman2009} (which is itself also a result following Kruskal's theorem \citep{kruskal1976more, kruskal1977three}), and the proof follows a similar construction provided in Corollary 5 from the same paper.\\
	
	We argue that the latent class model with $K$ components and each $Y_j$ having dimension $N_j$ is generically identifiable under the aforementioned bound.
	
	To simplify notation, let $M_j = 1+N_j$ and $M_* = \min_j M_j$.
	
	Let $K$ be fixed.  First, consider the case that $d=2\lceil \log_{M_*} K \rceil +1$.  Observe that
	$${M_*}^{\lceil \log_{M_*} K \rceil - 1} < K \leq {M_*}^{\lceil \log_{M_*} K\rceil}.$$
	
	Partition the set $\{M_j\}_{j\in[d]}$ into a singleton $P_3 = \{{M_*}\}$ and two sets $P_1$ and $P_2$ of equal size (each with cardinality $\lceil \log_{M_*} K\rceil$).  Then, we have
	\begin{align*}
		&\kappa_1 := \prod_{j\in P_1} M_j \geq M_*^{\lceil \log_{M_*} K\rceil}, \\
		&\kappa_2 := \prod_{j\in P_2} M_j \geq M_*^{\lceil \log_{M_*} K\rceil}, \\
		&\kappa_3 := \prod_{j\in P_3} M_j = M_*. 
	\end{align*}
	Then, it follows that $\min(K, \kappa_1) = \min(K, \kappa_2) = K$ and $\min(K, \kappa_3) \geq 2$.  Thus, we have
	$$\min(K, \kappa_1) + \min(K, \kappa_2) + \min(K, \kappa_3) \geq 2K + 2.$$  When $d > 2\lceil \log_{1+\min_j N_j} K\rceil +1$, one can partition $\{M_j\}_{j\in[d]}$ in a similar way such that $\kappa_1$ and $\kappa_2$ strictly increase.  Therefore, the previous inequality will still hold.  
	
	Now, finally, note that a binomial  model is a special case of a multinomial model and each $Y_j$ having dimension $N_j$.  
	So binomial product models are special cases of the latent class model of multinomial in Theorem~\ref{thm::allman}. 
	Generic identifiability of the larger model implies generic identifiability of the submodel.
\end{proof}

The following lemma is used to show how the generic identifiability of the ``no covariates'' model can imply the generic identifiability of the model with covariates under certain assumptions.

\begin{lemma}
	Suppose the logistic regression model is identifiable, and each mixture is distinct.  Then, the generic identifiability of the no covariate binomial product model implies the generic identifiability of the mixture of binomial product experts.
	\label{lemma:binomialproductexperts-identification}
\end{lemma}
\begin{proof}[Proof of Lemma \ref{lemma:binomialproductexperts-identification}]
	
	As each mixture is distinct, this implies that the model
	$$\sum_{k=1}^K w_k(x;\beta_k) p_k(y;\theta_k)$$
	is constructed such that $\theta_a \neq \theta_b$ for $a\neq b$.  We organize the parameters into one $\balpha = (\bbeta, \btheta)=(\beta_1,\ldots,\beta_K,\theta_1,\ldots,\theta_K)$ and $\alpha_k = (\beta_k,\theta_k)$.  For $\sigma \in S_K$, where $S_K$ is the permutation group comprising all $K!$ mappings from $[K]$ to $[K]$,
	we use the notation $\bbeta_\sigma$ and $\btheta_\sigma$ to denote the reordered tuples of $\bbeta$ and $\btheta$, respectively, according to the permutation $\sigma$.
	
	Now, consider a given covariate $x\in\mathcal{X}$, and suppose that
	$$p(y|x;\balpha) = \sum_{k=1}^K w_k(x;\beta_k) p_k(y;\theta_k) = \sum_{k=1}^K w_k(x;\beta'_k) p_k(y;\theta'_k) = p(y|x;\balpha').$$
	
	By generic identification of the no covariate binomial product model up to permutation, we have
	$$\btheta = \btheta'_{\sigma_x}$$
	for some permutation $\sigma_x \in S_K$.  This permutation is indexed by $x$ because this permutation may depend on the covariate $x$.
	
	We now argue that this permutation is the same for any $x$, thereby implying that the entire mixture of binomial products model is unique up to permutation.  We proceed using a proof by contradiction.  Let $x_1$ and $x_2$ be the covariates of two distinct observations with corresponding permutations $\sigma_{x_1}$ and $\sigma_{x_2}$ such that $\sigma_{x_1} \not\equiv \sigma_{x_2}$.  Without loss of generality, assume that $\sigma_{x_1}$ is the identity permutation.  So, for $x=x_1$, we have
	\begin{equation}
		\theta_k = \theta_k' \quad \forall k\in [K] \label{eq:identity-perm}
	\end{equation}
	by the identity permutation.  For $x= x_2$, we have
	$$\theta_k = \theta_{\sigma_{x_2}(k)}' \quad \forall k\in[K].$$
	Since $\sigma_{x_2}$ is not the identity permutation, there exists $g \in [K]$ such that $\sigma_{x_2}(g) \neq g$.  Combining this fact with equation \eqref{eq:identity-perm} via transitivity, we have
	$$\theta_g' = \theta_g = \theta_{\sigma_{x_2}(g)}',$$
	but this violates the distinct mixture assumption of $\theta_a \neq \theta_b$ for $a\neq b$.  Thus, we have a contradiction, and the permutation $\sigma_x$ must be invariant to the covariate $x$.  Finally, since the logistic regression model is identifiable, the mixture of binomial product experts model is generically identified up to permutation.
\end{proof}

We are now ready to present the proof of Proposition \ref{prop:identif} (this proves the generic identifiability of our mixture of binomial product experts model), which is a synthesis of results from the previous two lemmas.

\begin{proof}[Proof of Proposition \ref{prop:identif}]
	Since the inequality $d\geq 2\lceil \log_{1+\min_j N_j} K\rceil +1$ is satisfied by Assumption \ref{assump:inequality}, we can invoke Lemma \ref{lemma:nocovariate-identification}.  This implies that the model without covariates is generically identified.
	
	Then, under Assumptions \ref{assump:distinctmix} and \ref{assump:identified-logistic}, we can invoke Lemma \ref{lemma:binomialproductexperts-identification}.  Thus, the model with covariates is also generically identified, thereby completing the proof.
\end{proof}

	\subsection{Validity of the Bootstrap}
	
	\label{appendix:bootstrap}
	
	In this subsection, we discuss the validity of the bootstrap for this MLE estimator.  Before we discuss the theoretical results for the bootstrap, we state the Berry-Esseen bound, which will be useful in the argument.
	
	\begin{lemma}[Berry-Esseen bound]
		Suppose $Z_1,Z_2,\ldots, Z_n$ are i.i.d. random variables such that $\E[Z_i] = \mu$, $\E[Z_i^2] = \sigma^2 > 0$, and $\E[|Z_i|^3] < \infty$.  Then, there exists a constant $C>0$ such that
		$$\sup_z|P(\sqrt{n}(\bar{Z}_n-\mu) \leq z) - \Phi(z)|) \leq \frac{C\E[|Z_i|^3]}{\sigma^3\sqrt{n}}.$$
	\end{lemma}
	
	To demonstrate the validity of the bootstrap, we want to argue that the bootstrap distribution is asymptotically equivalent to the sampling distribution of the estimator.  The main idea of the proof is to exploit the fact that the estimator $(\hat{\bbeta}, \hat{\btheta})$ has an asymptotically linear form (that is, the MLE behaves much like a sample mean asymptotically) and converges jointly to a multivariate normal centered at $(\bbeta, \btheta)$ at $\sqrt{n}$-rate.
	
	
	Define the statistic $\hat{T}_n := T_n(\bX_{1:n},\bY_{1:n}) = \sqrt{n}\cdot(\frac{1}{n}\sum_{i=1}^n \ell(\bbeta, \btheta ; \bX_i, \bY_i) - \E[\ell(\bbeta, \btheta ; \bX, \bY)])$ as the centered and rescaled log-likelihood function.  Let $P_*$ denote the conditional law taking the data $\bX_{1:n},\bY_{1:n}$ fixed and $P$ be the law under the true generating distribution.
	
	More precisely, we wish to show the following uniform bound
	$$\sup_t |P_*(T_n^* \leq  t | \bX_{1:n},\bY_{1:n}) - P(\hat{T}_n \leq t)| \stackrel{p}{\to} 0.$$
	We perform the following decomposition via the triangle inequality
	\begin{align*}
		&\sup_t|P_*(T_n^* \leq t | \bX_{1:n},\bY_{1:n}) - P(\hat{T}_n \leq t)| \leq \\
		&\underbrace{\sup_t|P_*(T_n^* \leq t | \bX_{1:n},\bY_{1:n}) - \Phi(t;\hat{\sigma})|}_{(A)} + \underbrace{\sup_t|\Phi(t;\hat{\sigma}) - \Phi(t;\sigma)|}_{(B)} + \underbrace{\sup_t|P(\hat{T}_n\leq t) - \Phi(t;\sigma)|}_{(C)}.
	\end{align*}
	
	Observe by smoothness of the normal CDF and the fact that $\hat{\sigma} \stackrel{p}{\to} \sigma$, Term (B) is $O_p(n^{-1/2})$.  Next, Term (C) is constructed as the difference between the sampling distribution of $\hat{T}_n$ and the normal CDF.  The sampling distribution is simply constructed via the MLE, so there is no asymptotic bias.  Therefore, by the Berry-Esseen bound, Term (C) is $O_p(n^{-1/2})$.
	
	The more challenging term to control is Term (A), which we will accomplish via another application of the Berry-Esseen bound.  Under the law $P_*$, the data $\bX_{1:n},\bY_{1:n}$ are fixed, so for some $K\in\mathbb{R}$, we have
	\begin{align*}
		\sup_t|P_*(T_n^* \leq t | \bX_{1:n},\bY_{1:n}) - \Phi(t;\hat{\sigma})| &\leq \frac{K \frac{1}{n}\sum_{i=1}^n |\ell(\bbeta, \btheta ; \bX_i, \bY_i)|^3}{\hat{\sigma}^3 \sqrt{n}} \\
		&\stackrel{p}{\to} \frac{K \E[|\ell(\bbeta, \btheta ; \bX_i, \bY_i)|^3]}{\sigma^3 \sqrt{n}} \\
		&= O_p(n^{-1/2}).
	\end{align*}
	where the last convergence follows from the convergence of $\frac{1}{n}\sum_{i=1}^n |\ell(\bbeta, \btheta ; \bX_i, \bY_i)|^3 \stackrel{\text{a.s.}}{\to} \E[|\ell(\bbeta, \btheta ; \bX_i, \bY_i)|^3] < \infty$ and $\hat{\sigma} \stackrel{p}{\to} \sigma$.
	
	Thus, we have 
	$$\sup_t|P_*(T_n^* \leq t | \bX_{1:n},\bY_{1:n}) - P(\hat{T}_n \leq t)| = O_p(n^{-1/2}).$$
	Since we have shown the result for the centered and rescaled log-likelihood function, the consistency of the bootstrap follows from standard M-estimation theory (as the MLE is an M-estimator), as desired.

\section{Further Comments on the Real Data Analysis}

\label{appendix:more-realdata}

\subsection{Reproducibility of the Real Data Analysis}

\label{appendix:reproducibility}

We report the distribution of the maximum log-likelihood estimates in Figure \ref{fig:loglikelihood_dist}.  We observe that most of the random initializations converge to a local maximum of the log-likelihood function, as evidenced by the mode of the distribution.  This distribution indicates how important it is to have a sufficiently large number of random initializations in order to properly explore the parameter space.  If the real data analysis was to repeated, we expect to obtain a similar histogram of log-likelihood values.

\begin{figure}[!b]
	\centering
	\includegraphics[width=11cm]{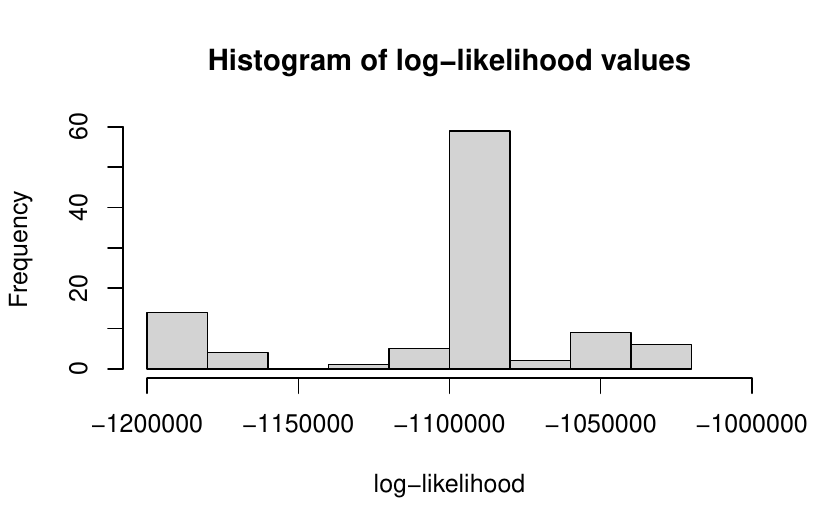}
	\caption{This figure is a histogram of the final log-likelihood values, reported over 100 random initializations, obtained using Algorithm \ref{alg:marEM} on the real data.}
	\label{fig:loglikelihood_dist}
\end{figure}

\clearpage

\newpage
\bibliographystyle{apacite}
\bibliography{references.bib}






%




\end{document}